\documentclass[lettersize,journal]{IEEEtran}
\usepackage{amsmath,amsfonts,amssymb,amsthm}
\usepackage{xcolor}
\usepackage{algorithm}
\usepackage{algpseudocode}
\usepackage{array}
\usepackage{textcomp}
\usepackage{stfloats}
\usepackage{url}
\usepackage{verbatim}
\usepackage{graphicx}
\usepackage{cite}
\usepackage{multirow}
\usepackage{caption}
\usepackage{subcaption}
\usepackage{hhline}
\usepackage{booktabs}

\def\BibTeX{{\rm B\kern-.05em{\sc i\kern-.025em b}\kern-.08em
    T\kern-.1667em\lower.7ex\hbox{E}\kern-.125emX}}
\usepackage{balance}

\newtheorem{theorem}{Theorem}
\newtheorem{lemma}{Lemma}

\def\b0{{\pmb{0}}} 

\def\ba{{\mathbf{a}}}  \def\bc{{\mathbf{c}}} 
  \def\bg{{\mathbf{g}}} 
   
 \def\bn{{\mathbf{n}}}  
   
\def\bu{{\mathbf{u}}}   
 \def\bz{{\mathbf{z}}}  

\def\bA{{\mathbf{A}}}   
  \def\bG{{\mathbf{G}}} \def\bH{{\mathbf{H}}}
\def\bI{{\mathbf{I}}}   
   
 \def\bR{{\mathbf{R}}}  
  \def\bW{{\mathbf{W}}} \def\bX{{\mathbf{X}}}
 \def\bZ{{\mathbf{Z}}}

\begin{document}
	
\title{Anomaly Detection-Based UE-Centric \\ Inter-Cell Interference Suppression}

\author{Kwonyeol Park,\thanks{K. Park is with the S.LSI Division, Samsung Electronics Company Ltd., Hwaseong-si, 18448, South Korea, and also with the School of Electrical Engineering, Korea Advanced Institute of Science and Technology, Daejeon 34141, South Korea (e-mail: kwon10.park@kaist.ac.kr).}
Hyuckjin Choi\thanks{H. Choi, B. Ko, M. Kim, G. Lee, and J. Choi are with the School of Electrical Engineering, Korea Advanced Institute of Science and Technology (e-mail: \{hugzin008, kobs0318, mjkim97, iee4432, junil\}@kaist.ac.kr).}, Beomsoo Ko, Minje Kim, Gyoseung Lee, \\Daecheol Kwon,\thanks{D. Kwon is with the S.LSI Division, Samsung Electronics Company Ltd., Hwaseong-si, 18448, South Korea, and also with the School of Electrical and Computer Engineering, University of California San Diego, La Jolla, USA (e-mail: dckwon@ucsd.edu).}
\thanks{H. Park, B. Kim, and M. Shin are with the S.LSI Division, Samsung Electronics Company Ltd., Hwaseong-si, 18448, South Korea (e-mail: \{hyunjae.park, bseung.kim, minho.shin\}@samsung.com).}Hyunjae Park, Byungseung Kim, Min-Ho Shin, and Junil Choi
}



\maketitle

\begin{abstract}
The increasing spectral reuse can cause significant performance degradation due to interference from neighboring cells. In such scenarios, developing effective interference suppression schemes is necessary to improve overall system performance.
To tackle this issue, we propose a novel user equipment-centric interference suppression scheme, which effectively detects inter-cell interference (ICI) and subsequently applies interference whitening to mitigate ICI. The proposed scheme, named Z-refined deep support vector data description, exploits a one-class classification-based anomaly detection technique. 
Numerical results verify that the proposed scheme outperforms various baselines in terms of interference detection performance with limited time or frequency resources for training and is comparable to the performance based on an ideal genie-aided interference suppression scheme.
Furthermore, we demonstrate through test equipment experiments using a commercial fifth-generation modem chipset that the proposed scheme shows performance improvements across various 3rd generation partnership project standard channel environments, including tapped delay line-A, -B, and -C models.
\end{abstract}


\section{Introduction}
Over the past decade, the unprecedented expansion of wireless communication systems has profoundly transformed our daily lives. 
Compared to previous systems such as fourth-generation (4G) long-term evolution (LTE), fifth-generation (5G) new radio (NR) aims to achieve far-reaching advancement to satisfy diverse demands, including high data rate, ultra-reliability, and low latency \cite{liberg2024introducing}. 
To meet the challenging requirements, 5G NR employs core technologies, e.g., utilization of higher frequency bands at frequency range 2 (FR2), extensive deployment of massive multi-beam antennas, increased network densification, and adoption of flexible numerology \cite{rappaport2017overview, larsson2014massive, kim2024meta, parkvall2017nr}. 
There are, however, many issues to implement these technologies. For example, network densification can lead to severe inter-cell interference (ICI) due to sharing of limited spectrum resources. ICI can be one of the major throughput-limiting factors and degrades the system performance.
Therefore, developing effective interference suppression schemes is essential for improving the performance of advanced wireless communication systems \cite{deng2017resource}.

Prior research has proposed several ICI suppression techniques for practical systems. For instance, 4G LTE included coordinated multi-point (CoMP), network-assisted interference cancellation and suppression (NAICS), and inter-cell interference coordination (ICIC) \cite{kuo2015new, bassoy2017coordinated, hamza2013survey}. Building upon these foundations, 5G NR introduced advanced interference mitigation schemes. One such scheme is beamforming coordination between gNodeBs (gNBs) in densely deployed millimeter-wave systems \cite{kim2017inter}. Another scheme is a semi-distributed dynamic ICIC, which is particularly suitable for heterogeneous environments such as two-tier femtocell networks \cite{padmaloshani2020semi}. 
While 5G NR employs various interference suppression schemes that facilitate deployment and ensure backward compatibility with legacy communication systems \cite{siddiqui2021interference}, these conventional schemes primarily rely on network-centric approaches, where multiple gNBs collaborate to mitigate interference. These network-centric interference suppression schemes face practical hurdles such as increased feedback overhead. To address these challenges in 5G systems, user equipment (UE)-centric interference suppression schemes, where the UE independently manages interference without the coordination of gNBs, have recently emerged as a viable solution in practical systems \cite{chaudhari2022machine}.

Whether network-centric or UE-centric, several interference suppression schemes can be considered as potential solutions, e.g., interference cancellation and interference alignment \cite{mikami2011inter, ntranos2015cellular}. However, acquiring channel state information (CSI) of ICI remains particularly challenging due to the dynamic nature of wireless channels and limited feedback capabilities in practical wireless communication systems \cite{lee2021downlink}. Therefore, these schemes encounter significant challenges in practical implementation. 
A simple, yet more efficient UE-centric interference suppression scheme could be interference whitening (IW)\cite{yu2013enhanced, kang2008pre, kim2022design}. IW treats ICI as colored Gaussian noise and whitens it together with the background noise. This process is straightforward and low-complexity since it employs a simple linear transformation without requiring exact knowledge of ICI characteristics. Unlike more complex interference suppression schemes, IW avoids decoding the interference signal individually, offering computational efficiency and easier implementation in practical systems.

While IW is a simple and effective UE-centric approach, its performance is well known to vary depending on the environment \cite{park2021reinforcement}. Specifically, IW is highly effective in interference-dominant environments, where ICI produces a consistent colored noise pattern, enabling better estimation of the whitening matrix even with limited reference signals. Conversely, in noise-dominant environments, IW may perform poorly due to inaccurate whitening matrix estimation, as the randomness of noise with limited reference signals leads to performance degradation. Therefore, the successful implementation of IW at the UE critically depends on accurately detecting the presence of ICI, as its performance is highly context-dependent.

In 5G NR, a UE can identify the presence of ICI utilizing two types of measurements. The first measurement is the reference signal received power (RSRP) from synchronization signal blocks (SSB) transmitted by neighboring gNBs. The second measurement is the CSI-interference measurement (CSI-IM) \cite{3gpp.38.214, 3gpp.38.215}, which uses strategically blanked resource elements. Based on these measurements, the UE can identify an interference-dominant environment when the RSRP of neighboring gNBs or the CSI-IM values exceed predefined thresholds. However, these measurements cannot detect ICI in the data region during actual data transmission, e.g., physical downlink shared channel (PDSCH). This limitation arises because the periodicity of both SSB and CSI-IM is typically much longer than the duration of data transmission \cite{3gpp.38.331}. Furthermore, since RSRP and CSI-IM utilize frequency resources different from those in PDSCH, accurately inferring ICI within the data region is challenging.

To overcome these limitations in 5G NR systems, we propose a novel UE-centric interference suppression scheme that leverages anomaly detection to effectively identify ICI and subsequently suppress it, as shown in Fig.~\ref{blockdiagram}. Anomaly detection is a technique designed to recognize patterns that deviate from expected behavior and is widely used in fields such as healthcare, military operations, and network security \cite{nassif2021machine, huszar2023toward, shaodiwang2023}. In the context of dynamic wireless environments, anomaly detection offers a promising solution for classifying the presence of ICI by identifying deviations from normal patterns.

Among various anomaly detection techniques, we focus on one-class classification (OCC) due to its structural simplicity and effectiveness \cite{ruff2018deep}.
Other types of anomaly detection techniques often involve complex architectures. For example, reconstruction-based anomaly detection techniques require either generator-discriminator pairs or encoder-decoder architectures \cite{venkataramanan2020attention, hou2021divide, ristea2022self}, while knowledge distillation-based anomaly detection techniques entail simultaneous training of teacher and student models \cite{bergmann2020uninformed, salehi2021multiresolution, deng2022anomaly}. In contrast, the OCC constructs a decision boundary for latent features of normal data using a single neural network \cite{liznerski2020explainable, yi2020patch}. This simplicity and computational efficiency make OCC a suitable approach for UE-side implementation, given the hardware limitations of the UE.
Building on these benefits, we propose a novel interference suppression scheme, named Z-refined deep support vector data description (ZRD-SVDD), which employs a deep learning-based SVDD for anomaly detection. The main contributions of this paper are summarized as follows:

\begin{itemize}
    \item We demonstrate that the proposed ZRD-SVDD outperforms conventional OCC-based anomaly detection techniques used as baselines, e.g., one-class support vector machine (OC-SVM) and $k$-nearest neighbor ($k$-NN), in terms of detection accuracy. This superior performance is validated using various metrics derived from the confusion matrix.

    \item Despite its simple structure, numerical results reveal that the proposed ZRD-SVDD-based IW scheme demonstrates remarkable performance. It not only outperforms baseline-based IW and conventional IW schemes but also effectively approaches the performance of the ideal IW scheme. This ideal scheme, analogous to a genie-aided system, assumes complete knowledge of ICI presence, which represents a performance upper bound. Notably, our proposed scheme achieves performance comparable to the ideal scheme while utilizing only a small fraction of time or frequency resources.

    \item To validate the practical applicability of our proposed scheme, we implement the proposed scheme on a commercial modem chipset, i.e., Exynos 5400, and conduct experiments using test equipment (TE), i.e., Anritsu MT8000A. We evaluate the performance across various channel environments, including representative 3rd generation partnership project (3GPP) channel models: tapped delay line (TDL)-A, -B, and -C. The experimental results confirm that the proposed ZRD-SVDD-based IW scheme outperforms conventional IW schemes in realistic scenarios.
\end{itemize}

The remainder of this paper is organized as follows. Section~II describes the system and channel models and briefly introduces the conventional IW scheme. In Section~III, we present our proposed interference suppression scheme based on the anomaly detection technique. In Section~IV, we provide results and discussion from both numerical simulation and TE experiments. Finally, Section~V concludes the paper.

\begin{figure}[!t]
    \centering
   \includegraphics[width=1\columnwidth]{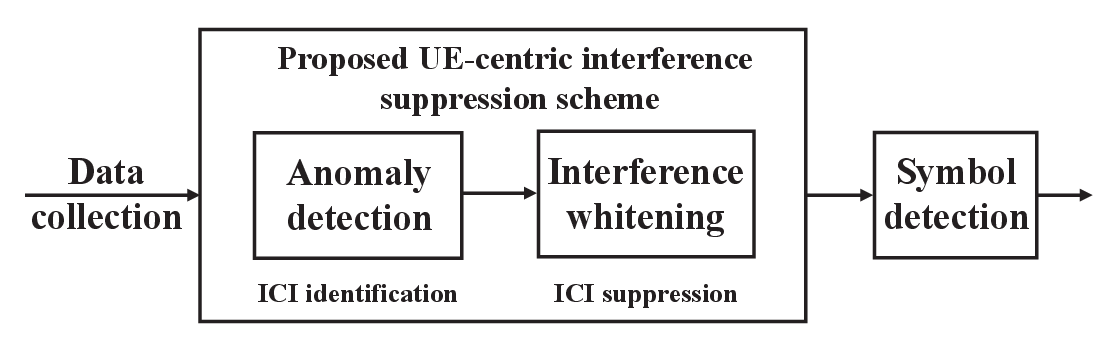}
    \caption{Block diagram of the proposed scheme.}
    \label{blockdiagram}
\end{figure}

\textbf{Notation:} 
Upper case and lower case boldface letters denote matrices and column vectors, respectively. 
The conjugate transpose, inverse, trace, and vectorization of a matrix $\mathbf{A}$ are denoted as $\mathbf{A}^{\mathrm{H}}$, $\mathbf{A}^{-1}$, $\mathrm{Tr}\left( \mathbf{A} \right)$, and $\mathrm{vec}\left(\mathbf{A}\right)$, respectively. The conjugate transpose of the matrix inverse $\mathbf{A}^{-1}$ is $\mathbf{A}^\mathrm{-H}$. The identity matrix and all-zero matrix of size $n \times n$ are represented by $\mathbf{I}_n$ and $\mathbf{0}_n$.
A circularly symmetric complex Gaussian distribution with mean vector $\mathbf{m}$ and covariance matrix $\mathbf{K}$ is represented using $\mathcal{CN}(\mathbf{m}, \mathbf{K})$. 
The set of all $m \times n$ matrices with complex-valued entries are symbolized by $\mathbb{C}^{m\times n}$, while $\mathbb{R}^n$ represents the set of $n$-dimensional real-valued vectors. The $\ell_2$-norm of a vector $\ba$, the spectral norm, and the Frobenius norm of a matrix $\bA$ are expressed as $\Vert \ba \Vert_2$, $\Vert \bA \Vert_2$, and $\lVert\bA\rVert_{\mathrm{F}}$, respectively. The expectation of a random variable $\bX$ is denoted by $\mathbb{E}[\bX]$.

\section{System Model and Problem Formulation}
\label{sec2}
In this section, we first introduce the system and channel model of interest and then discuss the concept and limitations of the conventional IW scheme, which motivates our proposed technique.
\subsection{System and channel model}
\label{sec2.A}
We consider a multi-cell downlink (DL) multiple-input multiple-output (MIMO) orthogonal frequency division multiplexing (OFDM) system as depicted in Fig.~\ref{scenario}, where each gNB serves its own UE. The gNBs and UEs have $M_{t}$ and $N_{r}$ uniform linear array (ULA) antennas, respectively. The DL received signal $\mathbf{y}_{m}$ of the UE~1 associated with the serving gNB in the $m$-th subcarrier is expressed as
\begin{align}\label{eq:system_model}
    \mathbf{y}_{m} &= \mathbf{H}_{m}\mathbf{x}_{m} +\sum_{n=1}^{N_{c}} \mathbf{G}_{n,m}\mathbf{z}_{n,m} + \mathbf{n}_{m} \\
    &= \mathbf{H}_{m}\mathbf{x}_{m} + \mathbf{G}_{m}\mathbf{z}_{m} + \mathbf{n}_{m},
\end{align} 
where $\mathbf{H}_{m}\in\mathbb{C}^{N_r\times M_t}$ is the desired channel between the serving gNB and the UE~1, and $\mathbf{G}_{n,m}\in\mathbb{C}^{N_r\times M_t}$ denotes the interference channel from the $n$-th neighboring gNB to UE~1. The aggregated interference channel $\mathbf{G}_{m}\in\mathbb{C}^{N_r\times M_t N_c}$ includes the interference channels from all $N_c$ neighboring gNBs and can be expressed as $\mathbf{G}_{m} = [\mathbf{G}_{1,m}, \cdots, \mathbf{G}_{N_c,m}]$. The transmit signal from the serving gNB is denoted as $\mathbf{x}_{m}\in\mathbb{C}^{M_t\times 1}$, satisfying $\mathbb{E}[\mathbf{x}_m \mathbf{x}_m^{\mathrm{H}}]=P_{\mathrm{S}} \mathbf{I}_{M_t}$, where $P_{\mathrm{S}}$ is the transmit power of the serving gNB. Likewise, the interference signal from the $n$-th gNB is represented by $\mathbf{z}_{n,m}\in\mathbb{C}^{M_t\times 1}$, satisfying $\mathbb{E}[\mathbf{z}_{n,m} \mathbf{z}_{n,m}^{\mathrm{H}}]=P_{\mathrm{I}} \mathbf{I}_{M_t}$, where $P_{\mathrm{I}}$ is the transmit power of the $n$-th gNB. The aggregated interference signal $\mathbf{z}_m \in \mathbb{C}^{M_tN_c \times1}$ is constructed by stacking the interference signals from the $N_c$ neighboring gNBs as $\mathbf{z}_m = [\mathbf{z}_{1,m}^{\mathrm{T}}, \cdots, \mathbf{z}_{N_c,m}^{\mathrm{T}}]^{\mathrm{T}}$. The Gaussian noise $\mathbf{n}_{m}\in\mathbb{C}^{N_r\times 1}$ follows $\mathcal{CN}(\b0_{N_r},\sigma_{m}^2\mathbf{I}_{N_r})$. The noise variance of the $m$-th subcarrier can be evaluated as $\sigma_{m}^2=N_0W$ where $N_0$ and $W$ are the noise spectral density and subcarrier spacing, respectively, assuming all subcarriers have the same bandwidth.

We consider the geometrical channel model for $\mathbf{H}_{m}$ and $\mathbf{G}_{m}$ \cite{liu2012cost}. The geometrical channel model is a path-based model, where the path defines the trajectory of the signal traveling from the gNBs to the UEs. A group of paths belongs to a cluster, which describes the channel scattering. To define the $m$-th subcarrier channel, we consider the continuous-time channel $\boldsymbol{\mathcal{H}}(t)\in\mathbb{C}^{N_r\times M_t}$ represented by
\begin{align}\label{eq:channel_model}
    &\boldsymbol{\mathcal{H}}(t) = \nonumber \\
&\quad\sum_{\ell=1}^{N_\ell}\sum_{p=1}^{N_p} \alpha_{\ell,p}e^{-j2\pi f_c t}\delta(t-\tau_{\ell,p}) \mathbf{a}_{N_r}(\theta_{\ell,p})\mathbf{a}_{M_t}^\mathrm{H}(\phi_{\ell,p}),
\end{align}
where $N_{\ell}$ is the number of clusters, $N_p$ is the number of paths in the $\ell$-th cluster, and $f_c$ denotes the carrier frequency. For the $p$-th path in the $\ell$-th cluster, $\alpha_{\ell,p}$, $\tau_{\ell,p}$, $\theta_{\ell,p}$, and $\phi_{\ell,p}$ represent the propagation loss, delay, angle of arrival (AoA), and angle of departure (AoD), respectively. The propagation loss $\alpha_{\ell,p}$ is modeled as 
\begin{align}    \alpha_{\ell,p}=\prod_{q=1}^{Q_p}\frac{\sqrt{P_0/4\pi}}{\eta(r_{\ell,p,q})},
\end{align}
where $P_0$ is the reference loss, and $\eta(r_{\ell,p,q})$ is the pathloss function depending on the propagation distance $r_{\ell,p,q}$ \cite{series2017guidelines, choi2023withray}. The number of linear trajectories in the $p$-th channel path is $Q_p$. The vector $\mathbf{a}_{N_r}(\cdot)$ represents an ${N_r}$-dimensional array steering vector as
\begin{align}
    \mathbf{a}_{N_r}(\theta_{\ell, p}) = [1\ e^{j\frac{2\pi d}{\lambda}\cos(\theta_{\ell, p})}\ \cdots\ e^{j\frac{2\pi d}{\lambda}({N_r}-1)\cos(\theta_{\ell, p})}]^{\mathrm{T}},
\end{align}
where $d$ and $\lambda$ are the antenna spacing and the wavelength, respectively. Similarly, $\ba_{M_t}\left( \cdot \right)$ denotes the $M_t$-dimensional array steering vector.

The $m$-th subcarrier channel results from the discrete Fourier transform (DFT), denoted as $\mathcal{F}\left\{\cdot\right\}$, of the channel in \eqref{eq:channel_model}. The desired channel $\mathbf{H}_{m}$ at the $m$-th subcarrier is expressed as
\begin{align}\label{eq:dft}
    &\mathbf{H}_{m}=\mathcal{F}\{\boldsymbol{\mathcal{H}}(t)\}=\int_{-\infty}^\infty \boldsymbol{\mathcal{H}}(t)e^{-j2\pi f_s \frac{m}{K}t}dt \nonumber\\
    &=\sum_{\ell=1}^{N_\ell}\sum_{p=1}^{N_p}\alpha_{\ell,p}e^{-j2\pi f_c \tau_{\ell,p}}
    e^{-j2\pi f_s\frac{m}{K}\tau_{\ell,p}}
    \mathbf{a}_{N_r}(\theta_{\ell,p})\mathbf{a}_{M_t}^\mathrm{H}(\phi_{\ell,p}),
\end{align}
where $f_s$ is the sampling frequency, and $K$ is the total number of subcarriers. 
The interference channel $\mathbf{G}_{m}$ at the $m$-th subcarrier can be similarly defined as in \eqref{eq:dft}.

\begin{figure}[!t]
    \centering
    \includegraphics[width=\columnwidth]{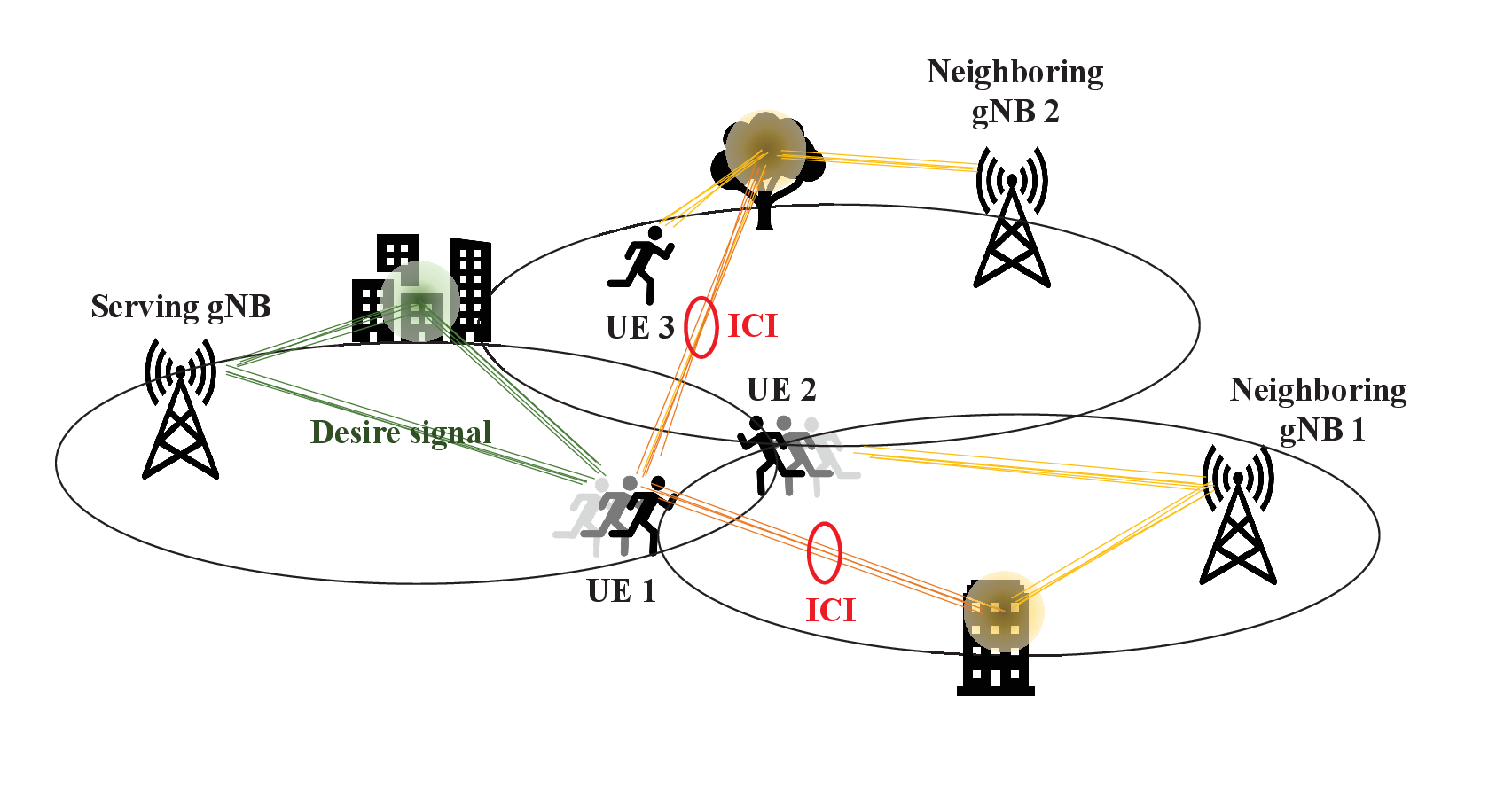}
    \caption{System model for the case of $N_c = 2$.}
    \label{scenario}
\end{figure}

\noindent{\textbf{\textit{Remark 1:}}} Defining clusters is beneficial for highlighting the temporal and spatial correlations of channels. Channel paths and clusters do not largely change over consecutive short-term sampling periods, resulting in channel variations that can be observed in actual measurement campaigns. 

\subsection{Conventional IW scheme and limitations}
\label{sec2.B}
When a UE is on cell edges, the colored noise, i.e., ICI, significantly degrades the decoding performance of the UE. To cope with this issue, the IW is employed to mitigate the effects of ICI and transform the characteristics of colored noise to approximately those of white noise. The process of the IW begins by characterizing the combined impacts of interference and noise on the received signal. At the $m$-th subcarrier, this combined signal is represented by the vector $\mathbf{u}_{m}$ as follows
\begin{align}\label{eq:NIV vector}
    \mathbf{u}_{m} &= \mathbf{G}_{m}\mathbf{z}_{m} + \mathbf{n}_{m}.
\end{align}
Under the assumptions of $\mathbb{E}[\mathbf{z}_m]=\boldsymbol{0}$ and $P_{\mathrm{I}}=1$, we can compute its covariance matrix as
\begin{align}\label{eq:covariance}
    \mathbf{R}_{m} = \mathbb{E} \left[\mathbf{u}_{m} \mathbf{u}_{m}^\mathrm{H} \right]  = \mathbf{G}_{m} \mathbf{G}_{m}^\mathrm{H} + \sigma_{m}^2 \mathbf{I}_{N_r}.
\end{align}
The covariance matrix can be factorized using the Cholesky decomposition, yielding $\mathbf{R}_{m} = \mathbf{L}_m \mathbf{L}_m^\mathrm{H}$, where $\mathbf{L}_m$ is a lower triangular matrix.
The whitened signal vector $\mathbf{y}_m^{\left(w\right)}$ is obtained by pre-multiplying the received signal vector $\mathbf{y}_m$ with the inverse of $\mathbf{L}_m$ \cite{sayed2020interference}, which results in the whitened covariance matrix $\mathbf{R}_{m}^{\left(w\right)}$ for $\mathbf{u}_{m}$  given by
\begin{align}\label{eq:whitendR}    \mathbf{R}_{m}^{\left(w\right)} 
    &= \mathbb{E} \left[\left(\mathbf{L}_m^{-1}\mathbf{u}_{m}\right)\left(\mathbf{L}_m^{-1}\mathbf{u}_{m}\right)^\mathrm{H} \right]\nonumber\\ 
    &= \mathbf{L}_m^{-1} \mathbb{E} \left[ \mathbf{u}_{m} \mathbf{u}_{m}^\mathrm{H} \right] \mathbf{L}_m^{-\mathrm{H}} = \mathbf{I}_{N_r}.
\end{align}

The key to achieving robust and optimal IW performance lies in accurately calculating the covariance matrix $\mathbf{R}_m$. However, the number of reference signals for estimating the effective ICI is usually limited in practical systems. Consequently, the sampled covariance tends to deviate from the true value of $\mathbf{R}_m$, which makes it challenging to obtain an accurate estimate. The following theorem provides insight into this estimation problem.

\begin{theorem}
    For a finite number of samples \(T_s\), the discrepancy between the sample covariance and the true covariance satisfies the following inequality
\begin{align}
\label{eq:Lemma1}
    &\mathbb{P}\left(  \left\Vert \hat{\bR}_m - \bR_m \right \Vert_2  < \epsilon  \right) \notag \ \\ 
&\quad\quad\quad \geq 1- 2N_r\exp\left(-\frac{\epsilon^2T_s^2/2}{\sigma_{\mathrm{F}}^2 + {2L_z\epsilon T_s}/{3}}\right),
\end{align}
where \( \hat{\bR}_m = \frac{1}{T_s} \sum_{t=1}^{T_s} \bu_{m,t} \bu_{m,t}^\mathrm{H} \) is the sample covariance with the $t$-th received interference plus noise signal $\bu_{m,t}$, and \(\epsilon\) denotes a design parameter representing the estimation error tolerance level. The parameters \(\sigma_{\mathrm{F}}^2\) and \( L_z\) are defined as follows
\begin{align}
\label{sigma_F_2_lemma}
\sigma_{\mathrm{F}}^2 = T_s &\sigma_m^2 \sqrt{N_r^2 \Vert \bg_m  \Vert_2^4 + 2 C_1N_r  \Vert \bg_m  \Vert_2^2 + C_1^2  N_r}, \notag \\
L_z &\approx 2 \sigma_m \sqrt{N_r} \Vert \bg_m \Vert_2 + \left( N_r - 1 \right) \sigma_m^2,
\end{align}
where \( \bg_m = \bG_m \bz_m\) and $C_{1} = \Vert \bg_m \Vert_2^2  + N_r \sigma_m^2 $.
\begin{proof}
     Refer to Appendix A.
\end{proof}
\end{theorem}

While the whitening process of the IW should theoretically enhance decoding performance by mitigating ICI, the estimation error may lead to suboptimal outcomes. To ensure an effective whitening process, it is crucial to minimize the difference between the sample and true covariance matrices. This can be achieved by increasing the number of samples, i.e., reference signals, or operating in high interference-to-noise ratio (INR) environments, as demonstrated in Lemmas~1 and~2.

\begin{lemma}
For any estimation error tolerance level \({\epsilon > 0} \), the discrepancy can be asymptotically smaller than \(\epsilon\) almost surely with a large number of samples \(T_s\).
\end{lemma}
\begin{proof}
As \( T_s\) approaches infinity to compute the lower bound in \eqref{eq:Lemma1}, we observe that
\begin{equation}
    \lim_{T_s  \rightarrow \infty } \exp\left(-\frac{\epsilon^2T_s^2/2}{\sigma_{\mathrm{F}}^2 + {2L_z\epsilon T_s}/{3}}\right) = 0.
\end{equation}
Since \(\epsilon\) can be independently selected regardless of \(T_s\), we obtain

\begin{equation}
\lim_{T_s \rightarrow \infty}  \mathbb{P}\left(  \left\Vert \hat{\bR}_m - \bR_m \right \Vert_2  < \epsilon  \right) = 1.
\label{discreeqaul0}
\end{equation}
\end{proof}

\begin{lemma}
For any estimation error tolerance level \({\epsilon > 0} \), the discrepancy can be asymptotically smaller than \(\epsilon\) almost surely in the high INR regime.
\end{lemma}
\begin{proof}
Using $\sigma^2_\mathrm{F}$ and $L_z$ from \eqref{sigma_F_2_lemma}, we can extend the lower bound in \eqref{eq:Lemma1} as
\begin{align}
\exp\left(-\frac{\epsilon^2T_s^2/2}{\sigma_{\mathrm{F}}^2 + {2L_z\epsilon T_s}/{3}}\right) =\exp\left(-\frac{\epsilon^2T_s^2/2}{\sigma_m^2 C_2 + {\frac{2}{3} \sigma_m } C_3}\right),
\end{align}
where $C_2 = T_s \sqrt{N_r^2 \Vert \bg_m  \Vert_2^4 + 2 C_1N_r  \Vert \bg_m  \Vert_2^2 + C_1^2  N_r}$ and $C_3= \epsilon T_s\left( 2\sqrt{N_r }\Vert \bg_m \Vert_2 +  \left( N_r - 1 \right) \sigma_m \right)$.

To examine the high INR regime, we consider the limit as $\sigma_m$ approaches zero as
\begin{align}
    \lim_{\sigma_m \rightarrow 0} \exp\left(-\frac{\epsilon^2T_s^2/2}{\sigma_m^2 C_2 + {\frac{2}{3} \sigma_m } C_3}\right) = 0. 
\end{align}
Consequently, the discrepancy becomes asymptotically smaller than $\epsilon$ in the high INR regime. 
\end{proof}

\begin{figure*}[!t]
    \centering
    \includegraphics[width=1.85\columnwidth]{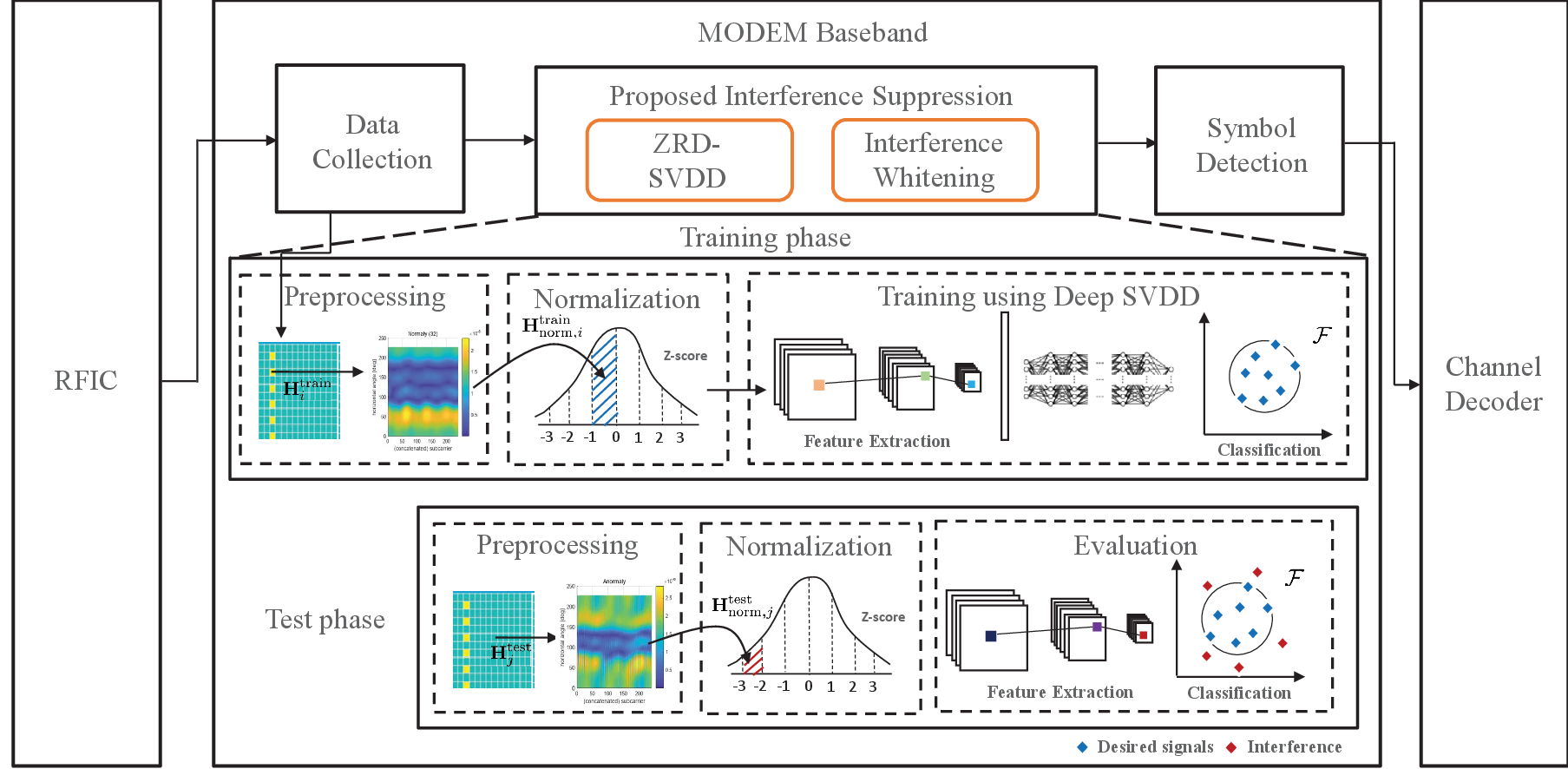}
    \caption{Framework for ZRD-SVDD.}
    \label{fig:proposedframework}
\end{figure*}

In systems with limited reference signals, such as 5G NR, the effectiveness of IW is closely related to the INR. To maximize performance gains, it is essential to apply IW adaptively based on the accurate detection of ICI. This requires overcoming the challenge of precisely identifying ICI, especially in scenarios with constrained time or frequency resources. To address this challenge, we employ a novel interference detection technique that models the ICI detection task as an anomaly detection problem. By leveraging this technique, the proposed scheme enables the adaptive application of IW, even under dynamic and resource-limited conditions. The following section introduces the proposed scheme, which is designed to tackle these challenges. 

\section{Proposed Scheme: Z-Refined Deep SVDD}
In this section, we propose a novel anomaly detection technique, named ZRD-SVDD, which is designed to detect ICI in various environments. ZRD-SVDD not only provides precise detection of strong interference but also improves sensitivity to subtle interference that is barely distinguishable from desired signals. Here, a desired signal is defined as a signal consisting of the signal from the serving gNB and noise, without any, or at least negligible, interference. In the following subsections, we first describe the components of ZRD-SVDD, and then present a detailed explanation of its training and testing phases.

\subsection{ZRD-SVDD components}
\label{sec3.A}
The proposed ZRD-SVDD consists of two main components: Z-score-based normalization and deep learning-based SVDD, as illustrated in Fig.~\ref{fig:proposedframework}. The first component employs the Z-score, a widely used method for data normalization in statistical anomaly detection. The Z-score standardizes the data by subtracting the mean and dividing by the standard deviation as given by
\begin{align} \label{eq:zscore}
    Z=\frac{X- \mu_X}{\sigma_X}
\end{align}
where $\mu_X$ and $\sigma_X$ are the mean and standard deviation of the data $X$, respectively.

The Z-score is computationally efficient, enabling rapid processing of large datasets, which is particularly advantageous for UE-centric approaches in wireless communication systems \cite{pang2019deep}.
By standardizing the data distribution, this normalization enhances interference detection accuracy, providing a consistent threshold for identifying outliers. Consequently, it effectively identifies subtle interference that might be obscured in the desired signal scale.

The second component is the deep learning-based SVDD. This component aims to find the smallest hypersphere in the feature space that encompasses desired signals, enabling the effective detection of complex interference patterns. Specifically, the neural network weights $\mathcal{W}=\left[ {\bW^1, \bW^2, ..., \bW^{N_{\mathrm{layer}}}} \right]$, with $N_{\mathrm{layer}}$ neural network layers, transform the training and test data samples into a new representation in a feature space through a non-linear transformation $\Omega(\cdot;\mathcal{W})$. This transformation enhances the separability between desired signals and anomaly signals such as ICI, by learning feature representations that maximize discrimination.

\subsection{Training phase of ZRD-SVDD} 
\label{sec3.B}
The training phase of the proposed ZRD-SVDD consists of four key steps: initialization for training, desired signal collection, Z-score-based data normalization, and learning the decision boundary through deep learning-based SVDD.

First, the UE leverages the existing 5G NR measurement including RSRP or CSI-IM to determine whether to learn the characteristics of the desired signals. When the RSRP from the serving gNB, $\mathrm{RSRP}_{\mathrm{S}}$, is stronger than that from the neighboring gNBs, $\mathrm{RSRP}_{\mathrm{N}}$, by a threshold for training $\rho_{\mathrm{tr}}$, i.e., $ \mathrm{RSRP}_{\mathrm{S}} - \mathrm{RSRP}_{\mathrm{N}}  \geq \rho_\mathrm{tr}$, the corresponding UE is considered to be in a noise-dominant environment. Additionally, CSI-IM provides another applicable metric for interference assessment. When CSI-IM measurement values fall below a predetermined threshold for training $\gamma_{\mathrm{tr}}$, i.e., $\mathrm{CSI\textendash IM} \leq \gamma_\mathrm{tr}$, the UE initiates the learning process to capture the characteristics of desired signals. This approach employs a semi-supervised learning framework, as it utilizes only desired signals during the training phase to model normal behavior, without requiring explicit labels for interference. 

In noise-dominant environments, a UE collects desired signals from the serving gNB using reference signals such as demodulation reference signal (DMRS). The collected desired signals are preprocessed by concatenating the signals across both angular and frequency domains to form the training data samples $\mathbf{H}^{\mathrm{train}}_{i}$ = $\mathrm{vec}\left( \left[ \mathbf{H}_m \right]_{m=1}^{N_f} \right), i=1,\cdots, N_t$, where $\mathbf{H_m}$ is defined in \eqref{eq:dft}, $N_f$ represents the number of subcarriers, and $N_t$ denotes the number of training data samples in the dataset. These training data samples are normalized using \eqref{eq:zscore} as $\bH^\mathrm{train}_{\mathrm{norm}, i} ={(\bH_i^\mathrm{train} - \mu^{\mathrm{train}})}/{\sigma^{\mathrm{train}}}$, where $\bH^\mathrm{train}_{\mathrm{norm}, i}$ represents the $i$-th normalized training data sample, and $\mu^{\mathrm{train}}$ and $\sigma^{\mathrm{train}}$ refer to the sample mean and standard deviation of the training dataset, respectively.

Using the normalized training data samples $\mathbf{H}^{\mathrm{train}}_{\mathrm{norm},i}, i=1,\cdots, N_t$, the deep SVDD model is trained to effectively separate desired signals from ICI. This is achieved by iteratively adjusting the neural network weights $\mathcal{W}$ to minimize the distance between transformed data and the center of the hypersphere, which decides the boundary between desired signal and ICI as \cite{ruff2018deep}
\begin{align}\label{eq:deepSvddObj}
\min_{\mathcal{W}} \frac{1}{N_t} \sum_{i=1}^{N_t} \Vert \Omega\left( \bH^{\mathrm{train}}_i;\mathcal{W} \right)-\bc \Vert_2^2 +\frac{\zeta}{2} \sum_{\nu=1}^{N_\mathrm{layer}} \Vert \bW^{\nu} \Vert^2_F,
\end{align}
where $\zeta$ and $\bc$ stand for a regularization parameter and the center of the hypersphere in the feature space, respectively. The first term in~\eqref{eq:deepSvddObj} minimizes the distance between the transformed data and the hypersphere center, enabling the model to distinguish between desired signals and ICI. The second term is a regularization term to prevent overfitting and improve generalization, which allows the model to better handle unseen data patterns. Consequently, this component ensures consistent performance across various interference scenarios, including cell center regions and high interference conditions often encountered at cell edges or in dense network deployments.

\begin{algorithm}[t]
\caption{ZRD-SVDD training phase}
\label{algorithm1}
\begin{algorithmic}[1]
\If{$ \mathrm{RSRP}_{\mathrm{S}} - \mathrm{RSRP}_{\mathrm{N}}  \geq \rho_\mathrm{tr}$ or $\mathrm{CSI\textendash IM} \leq \gamma_\mathrm{tr}$}
    \State Initialize $\bc$, $\mathcal{W}$, $\mu^{\mathrm{train}}$, and $\sigma^{\mathrm{train}}$
    \For{$i = 1$ to $N_t$}
        \State Collect $\bH_i^\mathrm{train}$
        \State Update $\mu^{\mathrm{train}}$ and $\sigma^{\mathrm{train}}$ using $\bH_i^\mathrm{train}$
    \EndFor
    \For{$\mathrm{epoch} = 1$ to $100$}
        \For{$i = 1$ to $N_t$}
            \State Generate $\bH^\mathrm{train}_{\mathrm{norm}, i}$ using \eqref{eq:zscore}
            \State Update $\mathcal{W}$ by solving \eqref{eq:deepSvddObj}
        \EndFor
        \State Update $\bc$ as the mean of $\Omega(\bH^\mathrm{train}_{\mathrm{norm}, i}; \mathcal{W})$ 
        \State Save the updated weights: $\mathcal{W}^{\star} \gets \mathcal{W}$
    \EndFor
    \State Set $\bc_{\mathrm{tr}} \gets \bc$
    \State \Return $\mathcal{W}^{\star}, \bc_{\mathrm{tr}}$
\EndIf
\end{algorithmic}
\end{algorithm}

This training process terminates when a predefined number of epochs is reached.
During the training process, the center of the hypersphere $\bc$ is initialized to a predefined value and iteratively updated until convergence. Upon the completion of training, it is finalized as $\bc_{\mathrm{tr}}$, representing the center of the learned hypersphere for distinguishing desired signals from ICI. Simultaneously, the optimal neural network weights $\mathcal{W}^{\star}$ are obtained by minimizing the objective function in \eqref{eq:deepSvddObj}. These outputs are crucial for the test phase, where they are used to determine whether test data correspond to desired signals or ICI.
The entire training phase is summarized in Algorithm~\ref{algorithm1}. Notably, this training phase can be performed periodically or triggered by significant changes in network conditions, enabling the model to adapt to evolving channel environments and maintain its effectiveness in detecting ICI.

\begin{algorithm}[t]
\caption{ZRD-SVDD test phase}
\label{algorithm2}
\begin{algorithmic}[1]
\If{$\vert \mathrm{RSRP}_{\mathrm{S}} - \mathrm{RSRP}_{\mathrm{N}} \vert < \rho_\mathrm{te}$ or $\mathrm{CSI\textendash IM} > \gamma_\mathrm{te}$}
    \For{j = 1 to $N_v$}
    \State Collect $\bH^\mathrm{test}_j$
    \State Generate $\bH^\mathrm{test}_{\mathrm{norm},j}$ by $\bH^\mathrm{test}_j$ using \eqref{eq:zscore}
    \State Compute $f_{\mathrm{zrd}}$ using (\ref{eq:deepSvdddecsion})
        \If{$f_{\mathrm{zrd}}\left(\bH^\mathrm{test}_{\mathrm{norm},j} \right) > \Theta$}
            \State Enable IW for $\bH^\mathrm{test}_j$ using \eqref{eq:whitendR}
        \Else
            \State Disable IW for $\bH^\mathrm{test}_j$
        \EndIf
    \EndFor
\EndIf
\end{algorithmic}
\end{algorithm}

\subsection{Test phase of ZRD-SVDD} 
\label{sec3.C}
The test phase of the proposed ZRD-SVDD is initiated when a UE potentially encounters ICI. This phase consists of several key steps including initialization for testing, test data normalization using Z-score, ICI detection through the trained deep SVDD, and adaptive operation of the IW based on the presence of ICI. 

To initialize the test phase, the UE utilizes the same metrics as in the training phase, i.e.,  RSRP or CSI-IM. The test phase is triggered when the RSRP difference between the serving gNB and neighboring gNBs falls below a predefined threshold for testing $\rho_{\mathrm{te}}$, i.e., $\vert \mathrm{RSRP}_{\mathrm{S}} - \mathrm{RSRP}_{\mathrm{N}} \vert < \rho_\mathrm{te}$ or when the CSI-IM measurement exceeds the threshold for testing $\gamma_{\mathrm{te}}$, i.e., $\mathrm{CSI\textendash IM} > \gamma_\mathrm{te}$. These conditions indicate that the UE is located in regions susceptible to ICI, such as cell edges or within dense network deployments.

In such environments with potential interference, the UE obtains the $j$-th test data sample $\bH^{\mathrm{test}}_j, j=1,\cdots, N_v $ from the data region, e.g., PDSCH, where $N_v$ denotes the number of test data samples, which are typically utilized for validation purposes. Similar to the training phase, these samples undergo preprocessing. The preprocessed test data samples are then normalized using the sample mean $\mu^{\mathrm{test}}_j$ and the sample standard deviation $\sigma^{\mathrm{test}}_j$ for each test data sample.{\footnote{Unlike the training phase, where the sample mean and standard deviation are obtained over the entire dataset, the test phase computes these statistics for each $j$-th test data sample individually.} Using \eqref{eq:zscore}, the $j$-th test data sample $\bH^\mathrm{test}_j$ is transformed into the normalized test data sample $\bH^\mathrm{test}_{\mathrm{norm},j}$. When the test data sample is subject to significant ICI, the corresponding normalized test data sample tends to fall in the tails of the normal distribution due to statistical differences from desired signals. The trained deep SVDD model maps $\bH^{\mathrm{test}}_{\mathrm{norm}, j}$ to the feature space, emphasizing the difference between desired signal and ICI. The presence of ICI is determined through the following decision function
\begin{align}\label{eq:deepSvdddecsion}
f_{\mathrm{zrd}}\left( \bH^{\mathrm{test}}_{\mathrm{norm}, j} \right) =\Vert \Omega\left(\bH^{\mathrm{test}}_{\mathrm{norm}, j};\mathcal{W}^{\star} \right) - \bc_{\mathrm{tr}}  \Vert_2^2.
\end{align}
This function measures the squared Euclidean distance between the mapped test data sample and the hypersphere center, where a distance exceeding a predefined threshold $\Theta$ indicates the presence of ICI. 

Following the interference evaluation, the IW is enabled only if ICI is detected, and remains disabled in the absence of ICI. The entire test phase is summarized in Algorithm~\ref{algorithm2}.

\begin{figure}[!t]
    \centering
    \includegraphics[width=\columnwidth]{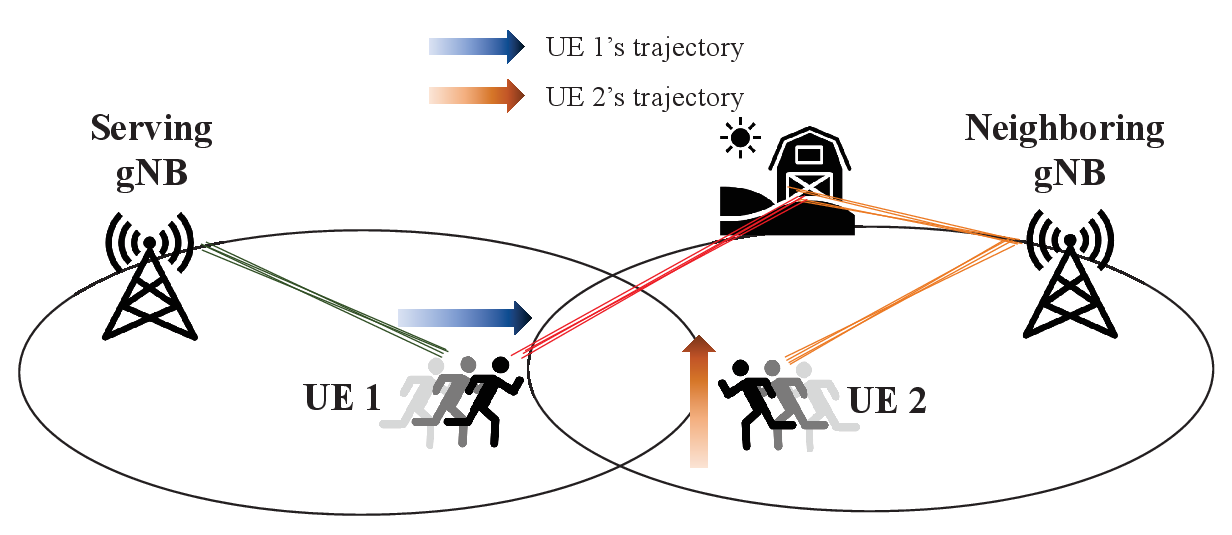}
    \caption{Experimental scenario for the case of $N_c = 1$.}
    \label{fig:experimentalScenario}
\end{figure}

\section{Performance Evaluation}
In this section, we evaluate the proposed ZRD-SVDD-based IW scheme through numerical simulations and TE experiments. Section~\ref{subsec: IV-A} details the experimental environments and datasets used for evaluation. Section~\ref{subsec: IV-B} introduces the baselines, and Section~\ref{subsec: IV-C} presents the complexity analysis including the proposed scheme and baselines. Section~\ref{subsec: IV-D} outlines the performance metrics employed for comparison. Finally, Section~\ref{subsec: IV-E} analyzes the experimental results and demonstrates the effectiveness of the proposed scheme.

\subsection{Experimental environments and datasets}
\label{subsec: IV-A}
\subsubsection{Numerical simulations} We consider the scenario in Fig.~\ref{fig:experimentalScenario}, where gNBs in each cell, with a radius of 40 $\mathrm{m}$, operate on the same frequency and provide service to their respective UE, with $N_c = 1$ neighboring cell considered.\footnote{For simplicity, this study considers a two-cell model with $N_c = 1$. However, we verified through other simulations that analogous trends are observed in a more general case with $N_c = 2$.} The signal received by the UE~1 from the neighboring gNB acts as ICI. In this scenario, the UE~1 and the UE~2 travel at a speed of 3 $\mathrm{m/s}$ in horizontal and vertical directions, respectively, with a total moving distance of 20 $\mathrm{m}$ for each UE. The UE~1 and the UE~2 move toward each other, meet at the midpoint, and then continue moving away from each other. The downlink channel data for both UEs are obtained every 0.1 $\mathrm{m}$, resulting in a total of 200 channel data samples at different locations for each UE.

We assume that both gNBs employ the singular value decomposition-based beamforming \cite{choi2024wmmse}. While this beamforming technique provides high throughput by concentrating the signal on the intended UE, it can lead to severe performance loss due to interference when the UE~2 moves toward the UE~1. For our simulations, the channel data samples are categorized into two distinct environments. Among 200 channel data samples, the indices 1-50 and 151-200 are considered to be in a noise-dominant environment, while the indices 51-150 represent an interference-dominant environment.
It is important to note that not all channel data samples with the indices 51-150 have been corrupted with the interference from the neighboring gNB since it may not exploit the same frequency resource with the serving gNB due to various factors, e.g., UE scheduling. To effectively simulate interference patterns within this environment, the indices 61-68, 86-90, and 121-123 are intentionally subjected to ICI. These samples should be classified as outliers that the UE needs to detect and mitigate. This assumption facilitates a realistic context for anomaly detection, including both interfered and non-interfered samples within the test dataset. Consequently, the channel data samples from the noise-dominant environment serve as the training dataset, while those from a potentially interference-prone environment are utilized as the test dataset.

\begin{table}[t]
\centering
    \caption{Parameters for simulations and TE experiments}
    \label{tab:SimulationParams}
\begin{tabular}{clcc}
\hline
\multicolumn{1}{c}{\textbf{Type}} & \multicolumn{1}{c}{\textbf{Parameter}} & \multicolumn{1}{c}{\textbf{Value}} & \multicolumn{1}{c}{\textbf{Unit}}  \\ \hline
& Carrier Frequency & 28 & $\mathrm{GHz}$  \\ 
& Sampling Frequency          &  122.88 & $\mathrm{MHz}$         \\ 
& Sub-carrier Spacing          &  120 & $\mathrm{kHz}$         \\ 
& OFDM Symbol          &  1024 & -        \\ 
Numerical& UE Mobility          &  3 & $\mathrm{m/s}$         \\ 
Simulation& Resource Blocks          &  20 & -         \\ 
& gNB Antennas            &  8 & -      \\ 
& UE Antennas           &  4 & -        \\ 
& Channel Clusters         &  4 & -        \\ 
& Paths per Cluster           &  5  & -       \\ \hline
& Carrier Frequency       &  3.3 & $\mathrm{GHz}$         \\ 
& Sub-carrier Spacing          &  30 & $\mathrm{kHz}$         \\
& Resource Blocks          &  273 & -         \\
TE& gNB Antennas           &  4 & -         \\
Experiment& UE Antennas          &  1 & -         \\
& MCS          &  3 & -         \\
& Rank          &  1 & -         \\
& Channel Models         &  TDL-A, -B, -C & -         \\ \hline
Common & Noise Spectral Density & -174 & $\mathrm{dBm/Hz}$ \\ \hline
\end{tabular}
\end{table}

Each channel data sample in the training and test datasets consists of angular and spectral domains. The angular domain corresponds to the number of receive antennas, while the spectral domain concatenates the reference signals within the allocated frequency resources.
To illustrate, we consider a scenario with 200 total channel samples in the time domain, where each of the aforementioned training and test datasets contains 100 samples, i.e., $N_t=N_v=100$ for each dataset. The gNB allocates 20 resource blocks (RBs), where each RB fully utilizes its frequency domain for reference signals, designated by $N_f = 12$ per RB. In this case, a single channel sample comprises $20 \times 12$ frequency resources per antenna, which are observed over 100 time instances per dataset. This structure captures the channel attributes across both frequency and time dimensions for each receive antenna. Simulation parameters for channel data generation are summarized in Table~\ref{tab:SimulationParams}.

\subsubsection{Experiments using TE}
We implement a realistic experimental setup using 5G NR TE to validate our proposed scheme. The experiment evaluates the performance of the ZRD-SVDD-based IW scheme under 5G NR conditions, comparing it with IW always-on and -off schemes.
This test environment consists of Anritsu MT8000A, a 5G base station emulator, and Samsung Exynos 5400 modem chipset for the UE as shown in Fig. \ref{fig:testEnvironment} \cite{MT8000A}.
We employ TDL channel models, specifically, TDL-A, -B, and -C \cite{3gpp.38.901}. These channel models simulate various multipath fading scenarios, which allow a comprehensive evaluation of the ZRD-SVDD-based IW scheme under different channel conditions.
Our experiment configurations are detailed in Table~\ref{tab:SimulationParams}. The modulation and coding scheme (MCS) is fixed to 3, and the rank\footnote{The definition of MCS and rank can be found in \cite{3gpp.38.214}.} is set to 1, emulating typical cell-edge conditions where the ICI is prevalent.
This experiment bridges the gap between theoretical simulations and practical implementations, providing insights into the performance of our proposed scheme in real-world 5G NR environments.

\begin{figure}[!t]
    \centering
    \includegraphics[width=.90\columnwidth]{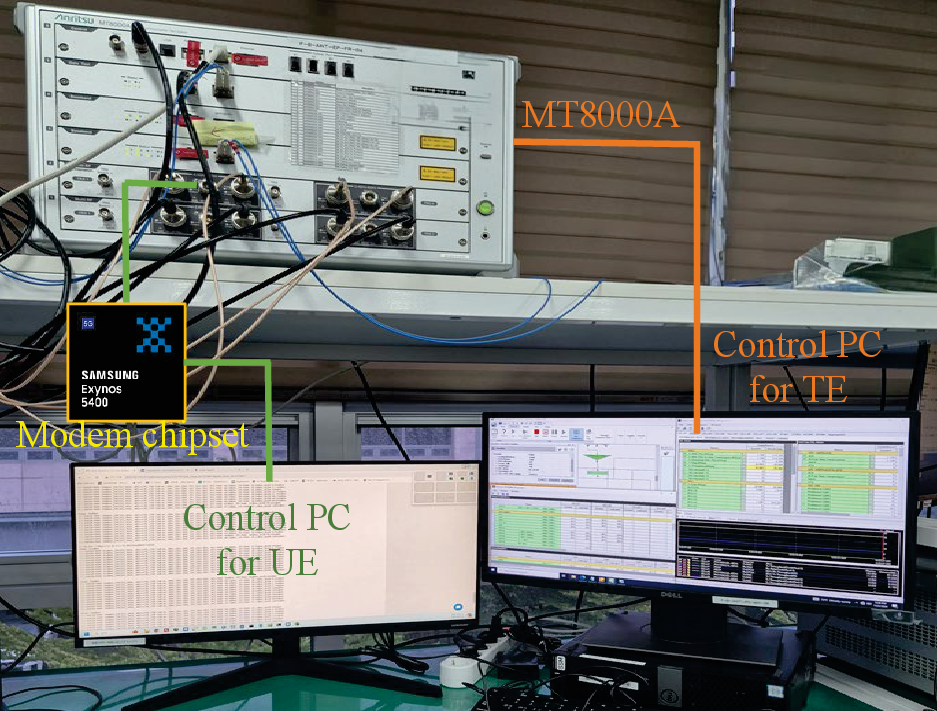}
    \caption{Environment for the TE experiments.}
    \label{fig:testEnvironment}
\end{figure}

\subsection{Baselines}
\label{subsec: IV-B}
To compare the performance of the proposed ZRD-SVDD, we employ OC-SVM and $k$-NN as baselines, which are well-established OCC algorithms commonly used for anomaly detection. OC-SVM constructs a hyperplane to separate normal data from anomalies in a high-dimensional feature space \cite{scholkopf2001estimating}. While primarily designed for multi-class classification, $k$-NN can be adapted for one-class classification in anomaly detection tasks by identifying anomalies based on the local density of data points in the feature space \cite{khan2018relationship}. In the context of $k$-NN, a small value of $k$ may lead to false alarm due to minor variations, whereas a large value of $k$ improves robustness by averaging across more samples.

Additionally, we establish baselines for the proposed ZRD-SVDD-based IW using conventional IW schemes, including IW always-on, IW always-off, and ideal genie-aided IW. The IW always-on scheme applies the whitening process regardless of ICI, while the IW always-off scheme consistently disables IW under all conditions. The genie-aided IW scheme assumes exact knowledge of whether ICI is present or not and adapts the operation of IW accordingly. This scheme serves as the performance upper bound.

\subsection{Complexity analysis}
\label{subsec: IV-C}

The training complexity of OC-SVM is dominated by solving a quadratic programming problem to construct the decision boundary, leading to a complexity of $\mathcal{O}(N_{t}^{3})$ \cite{scholkopf2001estimating}. During testing, OC-SVM computes the distance of each test sample from the hypersphere center, resulting in a test complexity of $\mathcal{O}(N_{v} \cdot F)$, where $F$ denotes the feature dimensionality of the data, i.e., $F \approx N_f * N_r$. 

For $k$-NN, the training phase involves storing the extracted features, with a complexity of $\mathcal{O}(N_{t} \cdot F)$. However, testing requires comparing each test sample with all training samples, yielding a test complexity of $\mathcal{O}(N_{t} \cdot N_{v} \cdot F)$, which makes $k$-NN less efficient for large-scale datasets.

ZRD-SVDD incorporates feature extraction and iterative optimization. Training complexity is $\mathcal{O}(E \cdot N_{t} (N_f \cdot N_r) + F )$, where $E$ is the number of epochs. Testing involves feature extraction and computing distances from the hypersphere center, resulting in $\mathcal{O}(N_{v} (N_f \cdot N_r) + F )$. Z-score normalization involves simple arithmetic operations and thus does not significantly impact the overall complexity. The proposed scheme offers a balanced approach to training and testing complexity. While OC-SVM suffers from cubic training complexity, making it impractical for large datasets, and $k$-NN requires exhaustive pairwise comparisons during testing, the proposed ZRD-SVDD achieves efficient training through iterative optimization and reduces testing complexity compared to $k$-NN, making it well-suited for large-scale and high-dimensional data.

In addition, we measure the time complexities of the proposed ZRD-SVDD and baseline schemes under the experimental environment described in Table~\ref{tab:experimental_environment}. The corresponding results are presented in Table~\ref{tab:time_comparison}. While the proposed ZRD-SVDD requires a slightly longer training phase due to the optimization of the hypersphere and neural network weights, it demonstrates a significant reduction in test time compared to the baselines. This result underscores its suitability for practical communication systems where rapid decision-making is critical.

\begin{table}[t]
\centering
\caption{Experimental environment for time complexity}
\begin{tabular}{|l|l|}
\hline
\textbf{Component}    & \textbf{Specification}                                                   \\ \hline
GPU Model             & NVIDIA GeForce RTX 3050                                                  \\ \hline
CUDA Version          & 11.3                                                                     \\ \hline
cuDNN Version         & 8302                                                                     \\ \hline
CPU                   & Intel Core™ i9-12900K                     \\ \hline
Cores / Threads       & 12 Cores, 20 Threads                                                     \\ \hline
RAM Size              & 31.84 GB                                                                \\ \hline
PyTorch Version       & 1.12.1                                                                   \\ \hline
\end{tabular}
\label{tab:experimental_environment}
\end{table}
\begin{table}[t]
\centering
\caption{Comparison of training and test times}
\label{tab:time_comparison}
\begin{tabular}{lcc}
\toprule
\textbf{Model}     & \textbf{Training Time (sec)} & \textbf{Test Time (sec)} \\
\midrule
\textbf{Proposed scheme} & 0.4812                 & 0.0352              \\
\textbf{OC-SVM}          & 0.0335                 & 0.1607              \\
\textbf{$k$-NN}            & 0.0020                 & 1.0018              \\
\bottomrule
\end{tabular}
\end{table}

\subsection{Performance metrics}
\label{subsec: IV-D}
The confusion matrix is a widely used metric for evaluating classification performance. It comprises the following sub-metrics, accounting for all possible combinations of predictions and actual values.

\begin{figure}[!t]
    \centering
    \includegraphics[width=\columnwidth]{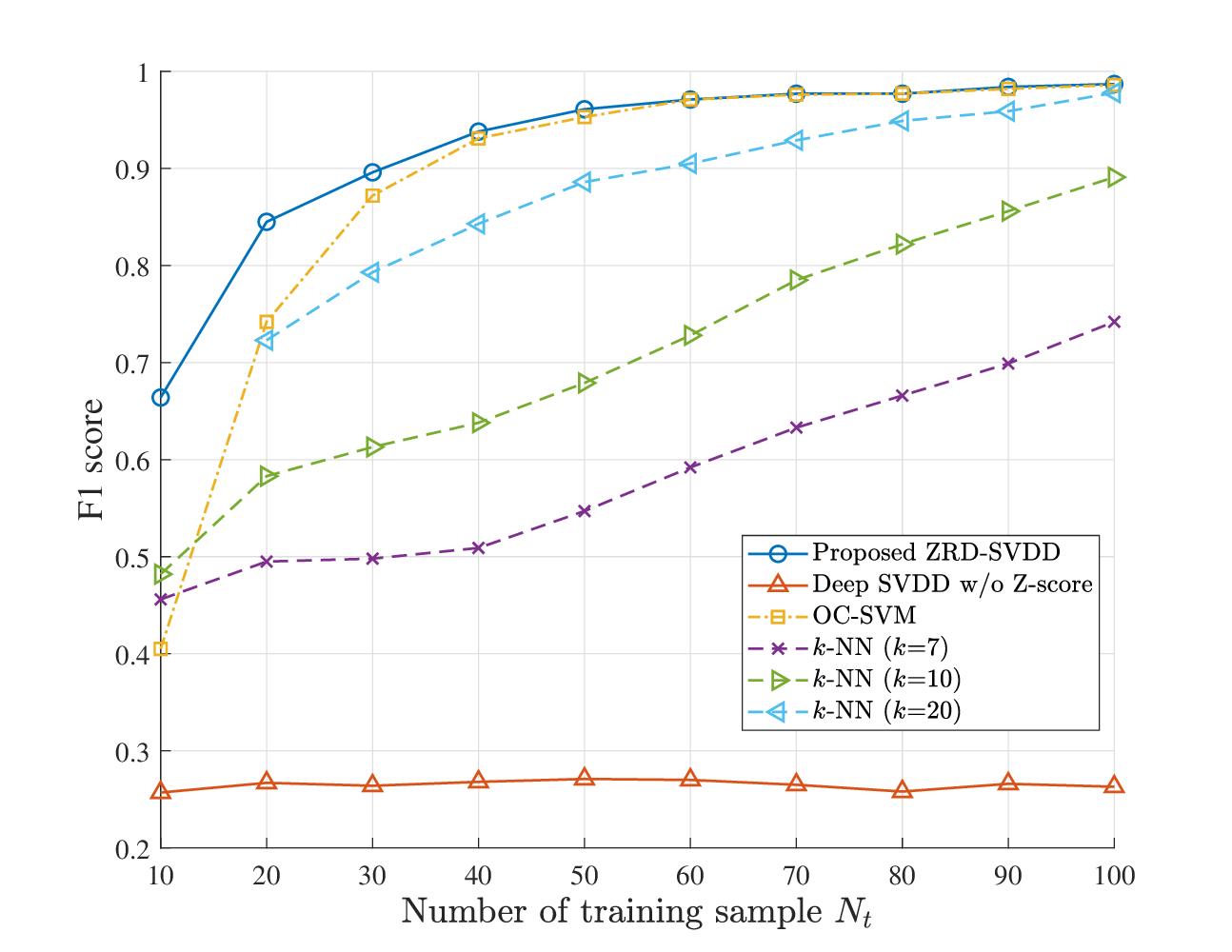}
    \caption{F1 score vs. number of training samples $N_t$ in time domain with $N_f = 12$ and $P_{\mathrm{S}} = P_{\mathrm{I}} = 30$ dBm.}
    \label{fig:5 F1 Score time}
\end{figure}

\begin{itemize}
    \item True positives (TP): These are cases where anomalies are accurately identified, which in this paper means correctly detecting ICI.
    \item False positives (FP): These are cases where normal conditions are incorrectly predicted as an anomaly. In the context of this paper, it refers to instances where ICI is falsely identified despite its absence, i.e., false alarms regarding ICI.
    \item False negatives (FN): These are cases where anomalies are incorrectly classified as normal. In this paper, this corresponds to the missed detection of ICI.
    \item True negatives (TN): These are cases where normal conditions are correctly classified as normal.
\end{itemize}
Rather than directly using these sub-metrics, we employ the following metrics for evaluation: sensitivity, precision, and F1 score.

The sensitivity measures the fraction of correctly identified anomalies as 
\begin{align}\label{eq:Sen}
\mathrm{Sensitivity} = \frac{\mathrm{TP}}{\mathrm{TP}+\mathrm{FN}}.
\end{align}
Since FP does not impact sensitivity, predicting all samples as anomalies would result in a perfect sensitivity score of one. For this reason, sensitivity is generally not used as a standalone metric.

The precision assesses the proportion of predicted anomalies that are true anomalies and is represented as
\begin{align}\label{eq:Precision}
\mathrm{Precision} = \frac{\mathrm{TP}}{\mathrm{TP}+\mathrm{FP}}.
\end{align}
Similar to sensitivity, precision is not solely used since it does not consider FN, which becomes problematic when all decisions are predicted as normal.

For accurate evaluation, the F1 score is calculated as the harmonic mean of sensitivity and precision. It is expressed as follows
\begin{align}\label{eq:f1score}
\mathrm{F1\enspace score} = \frac{2 \times \mathrm{Precision} \times \mathrm{Sensitivity}}{\mathrm{Precision}+\mathrm{Sensitivity}}.
\end{align}
The F1 score balances sensitivity and precision and is particularly useful for evaluating classification performance in scenarios with imbalanced distribution between normal and anomaly data.

In addition to the confusion matrix-derived metrics that evaluate detection performance, we employ three additional performance metrics to evaluate our proposed scheme: symbol error rate (SER), PDSCH block error rate (PDSCH BLER), and downlink throughput. SER measures the ratio of incorrectly decoded symbols to the total number of transmitted symbols \cite{choi2015quantized}. PDSCH BLER represents the ratio of erroneously received transport blocks to the total number of transmitted transport blocks. Downlink throughput quantifies the amount of data successfully transmitted from the gNB to the UE per unit of time, typically measured in bits per second, and serves as an indicator of overall system efficiency \cite{baknina2020adaptive}. While SER is primarily used in our numerical simulations, PDSCH BLER and downlink throughput are utilized in our TE experiments to evaluate real-world performance.
\begin{figure}[!t]
    \centering
    \includegraphics[width=\columnwidth]{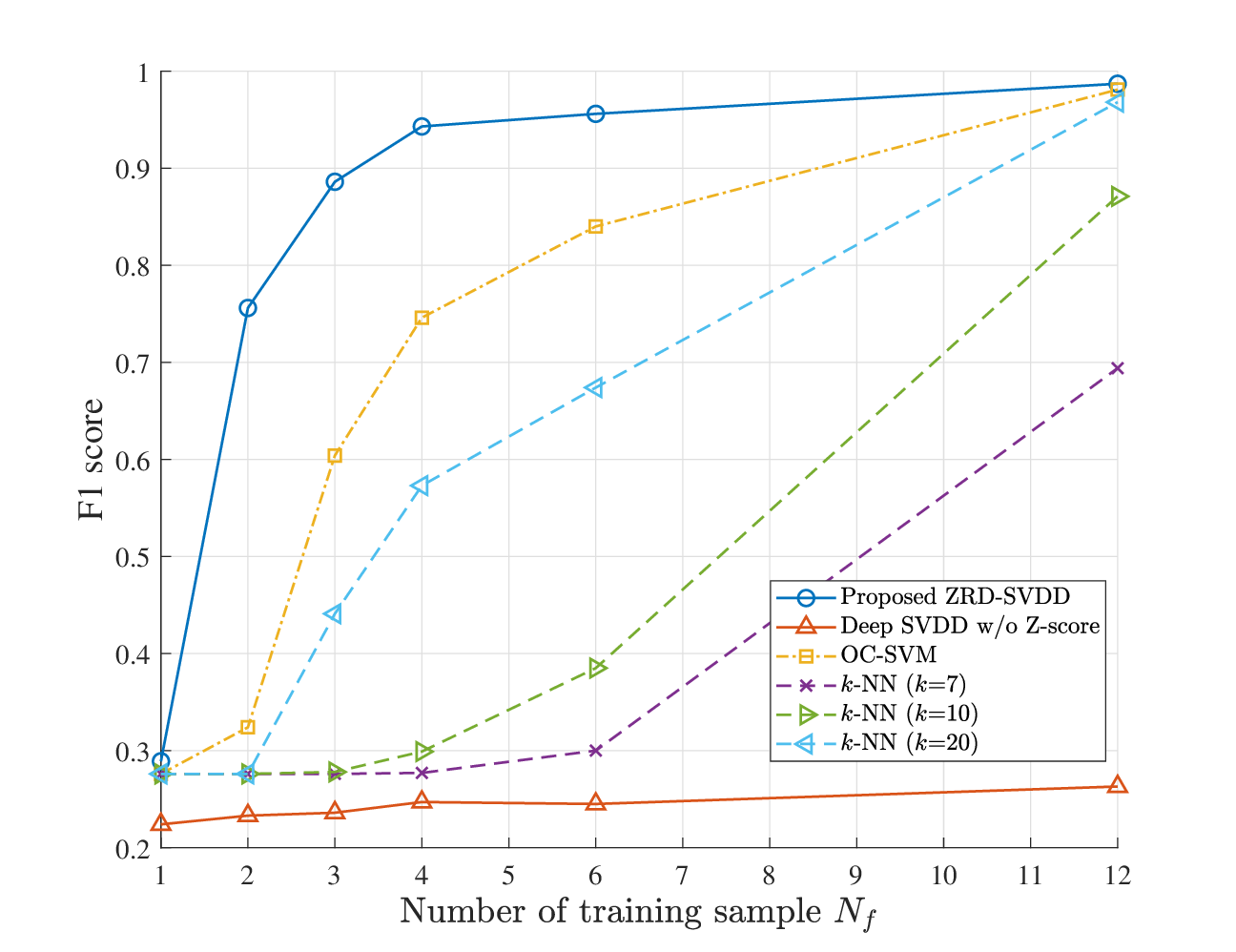}
    \caption{F1 score vs. number of training samples $N_f$ in frequency domain with $N_t = 100$ and $P_{\mathrm{S}} = P_{\mathrm{I}} = 30$ dBm.}    
    \label{fig:6 F1 Score Freq}
\end{figure}
\subsection{Experimental results}
\label{subsec: IV-E}
\label{sec5}

We evaluate the performance of the proposed ZRD-SVDD by varying three key parameters: the number of frequency resources $N_f$ and time resources $N_t$ used for training, as well as the transmit power of the serving gNB $P_{\mathrm{S}}$ and neighboring gNB $P_{\mathrm{I}}$. 

We first evaluate the interference detection performance of the proposed ZRD-SVDD under limited time resources to demonstrate its learning efficiency. The impact of the time domain resources on performance is investigated by progressively reducing the total training samples from 100 in decrements of 10. To prevent overfitting to a specific channel and ensure diversity in the training data, we randomly select subsets of samples from the original training dataset with 100 samples.

As shown in Fig. \ref{fig:5 F1 Score time}, the proposed ZRD-SVDD consistently outperforms the baselines in terms of the F1 score across various numbers of training samples in the time domain. Even when the number of training samples is reduced to 30, the ZRD-SVDD maintains the highest F1 score of approximately 0.84. This robustness to limited training samples is particularly significant in real-world applications where rapid decision-making is critical.
Furthermore, the performance of deep SVDD without Z-score normalization is significantly degraded compared to the proposed ZRD-SVDD. The lower performance remains consistent regardless of the number of training samples $N_t$. These results demonstrate that Z-score normalization enhances the accuracy of interference detection by establishing a more consistent threshold for identifying anomalies.
\begin{figure}[!t]
    \centering
    \includegraphics[width=\columnwidth]{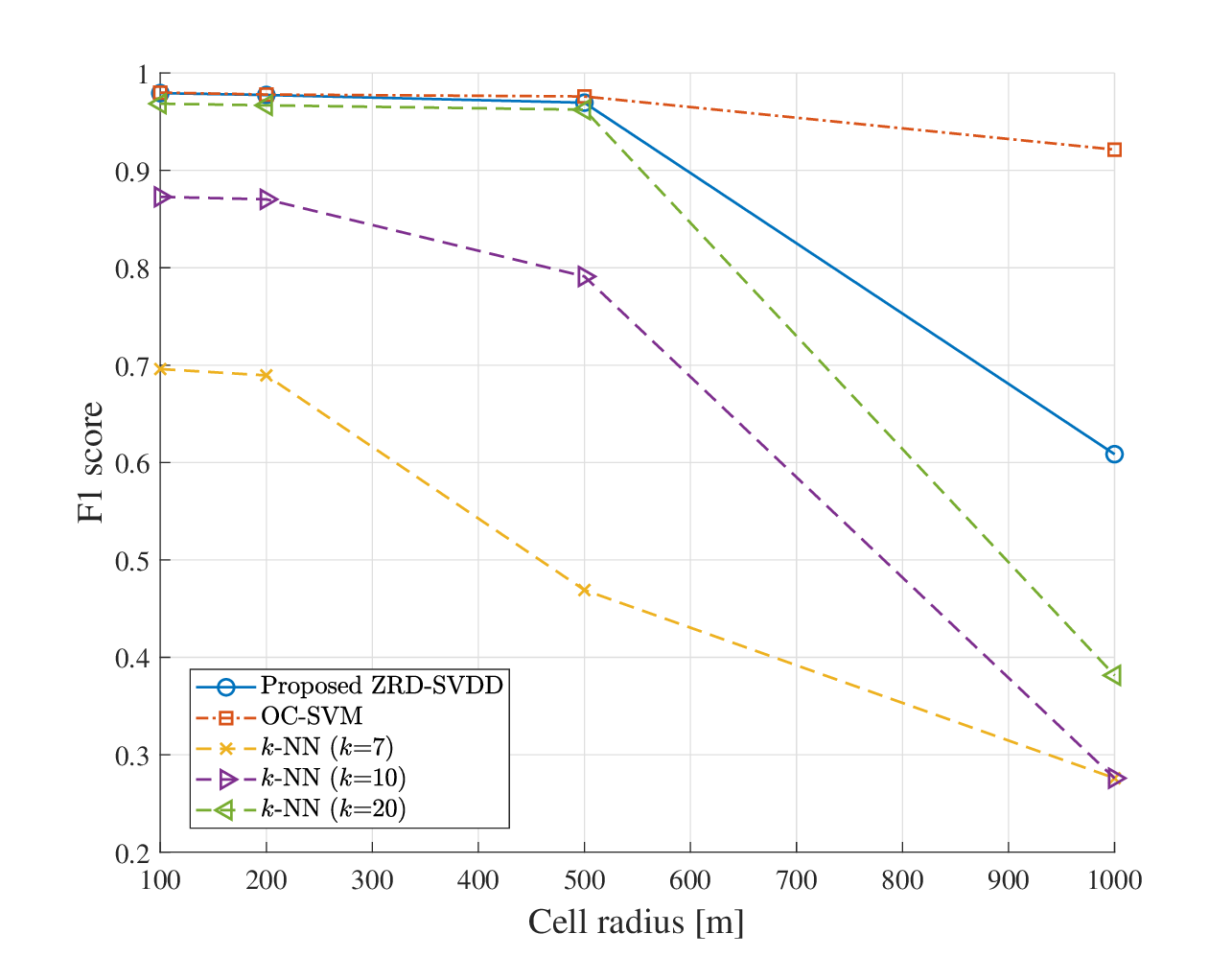}
    \caption{F1 score vs. cell radius with $N_t = 100$, $N_f = 12$ and $P_{\mathrm{S}} = P_{\mathrm{I}} = 46$ dBm.}
    \label{fig:f1score_tp}
\end{figure}

Fig.~\ref{fig:6 F1 Score Freq} evaluates the performance for different numbers of training samples $N_f$ in the frequency domain. This simulation considers 1, 2, 3, 4, 6, and 12 training samples per RB. In 5G NR, one RB consists of 12 resource elements (REs), thus utilizing 12 training samples corresponds to fully utilizing one RB from a frequency perspective. For configurations with fewer than 12 training samples, the training samples are evenly spaced at intervals of $12/N_f$ across the entire RB. Interpolation is employed for data acquisition between training samples, and extrapolation is applied for areas beyond these training samples.
In this simulation, the proposed ZRD-SVDD demonstrates superior performance compared to baselines across all frequency resource configurations. Even with a typical 5G NR configuration using 6 DMRS REs, the proposed ZRD-SVDD achieves an F1 score of approximately 0.96, showing at least 15\% improvement over the baselines. This result clearly shows that our proposed ZRD-SVDD scheme demonstrates the ability to provide accurate interference detection even with sparse frequency sampling.
\begin{figure}[!t]
    \centering
    \includegraphics[width=\columnwidth]{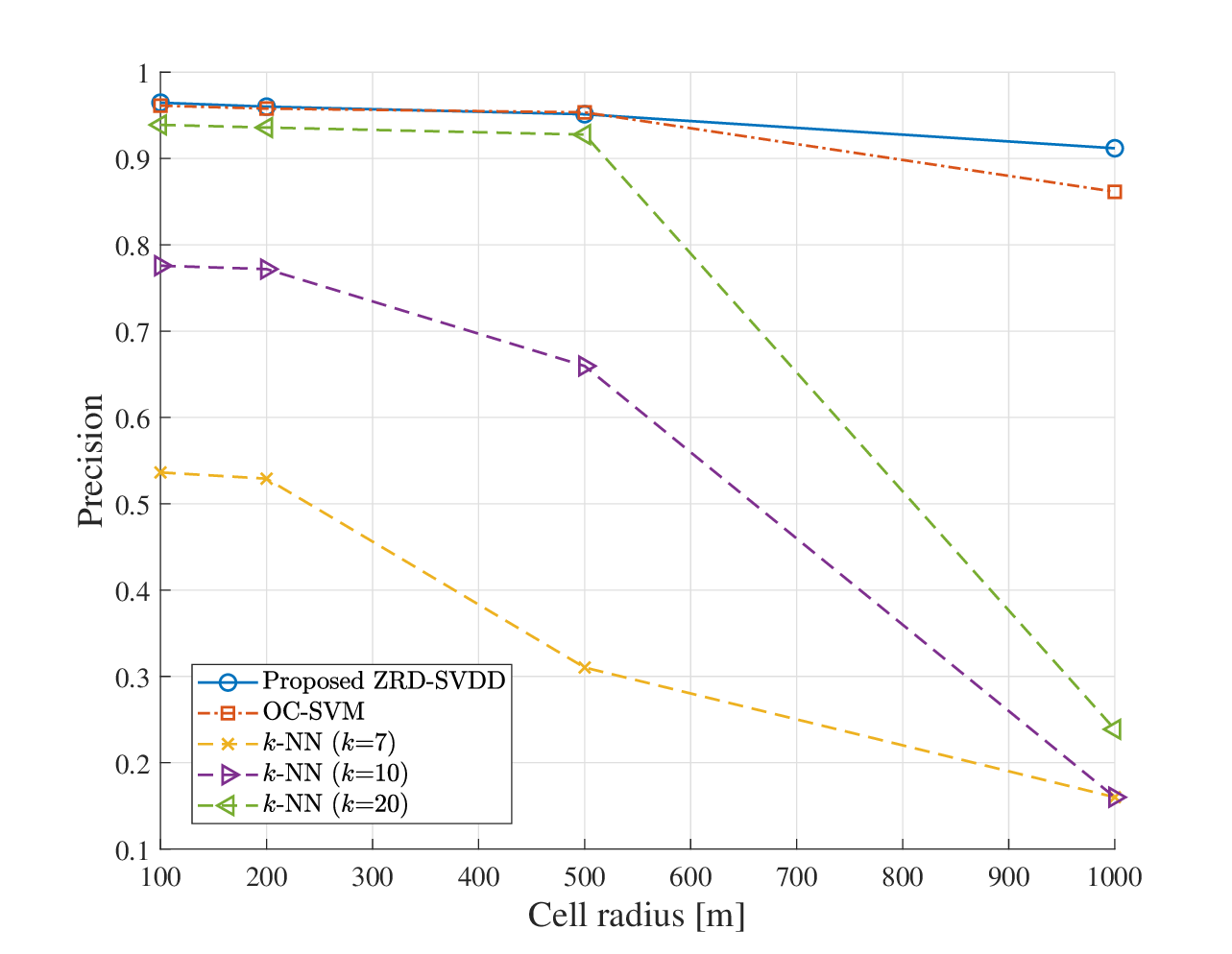}
    \caption{Precision vs. cell radius with $N_t = 100$, $N_f = 12$ and $P_{\mathrm{S}} = P_{\mathrm{I}} = 46$ dBm.}
    \label{fig:precision_tp}
\end{figure}

\begin{figure}[!t]
    \centering
    \includegraphics[width=\columnwidth]{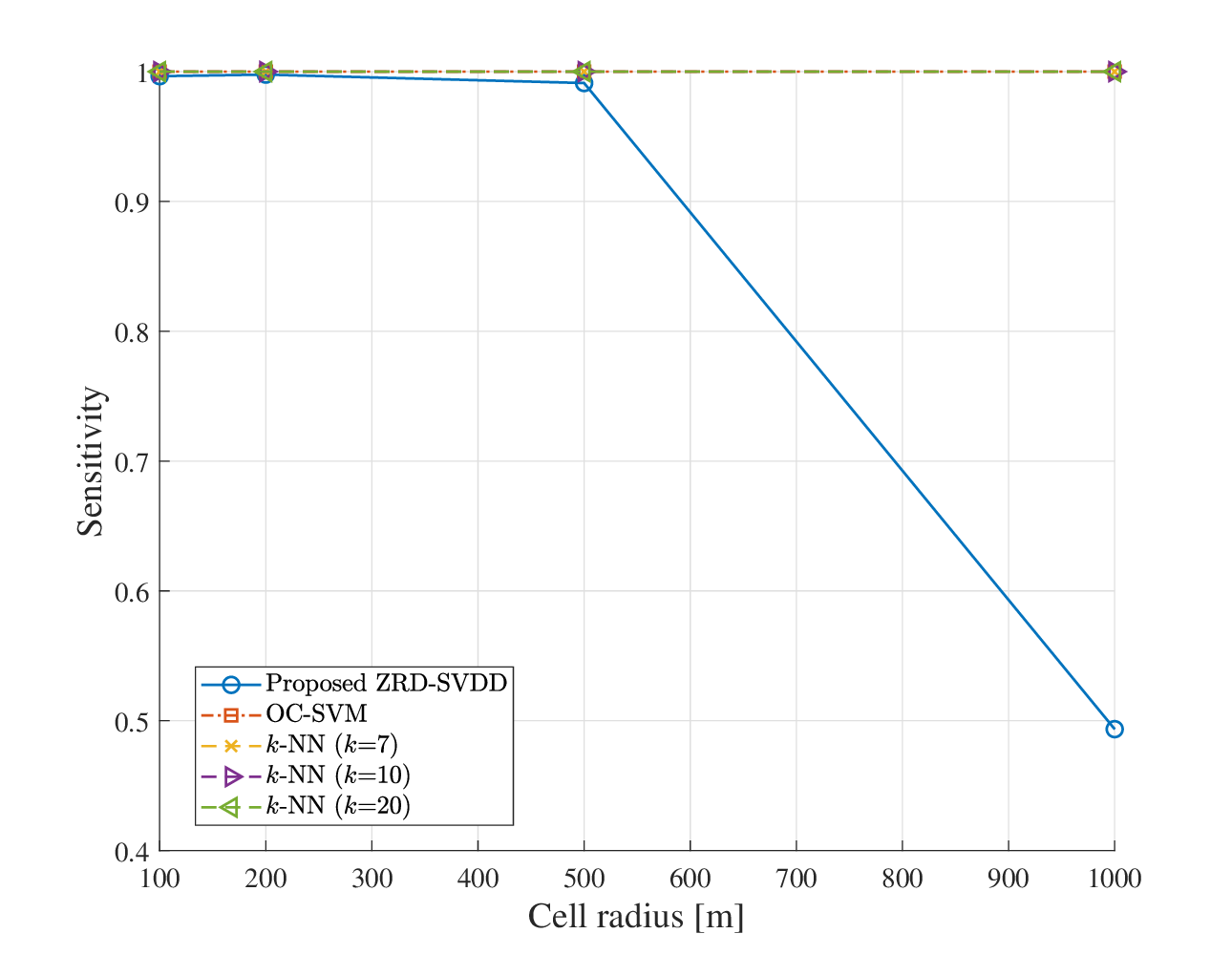}
    \caption{Sensitivity vs. cell radius with $N_t = 100$, $N_f = 12$ and $P_{\mathrm{S}} = P_{\mathrm{I}} = 46$ dBm.}
    \label{fig:sensitivity_tp}
\end{figure}
In the following three plots, we evaluate the performance of the proposed ZRD-SVDD by varying the cell radius from 100 $\mathrm{m}$ to 1 $\mathrm{km}$ with fixed transmit powers of 46 $\mathrm{dBm}$ for both the serving and neighboring gNBs. This configuration enables the assessment of performance across different cellular environments, ranging from dense small-cell deployments to expansive macro-cell coverage areas. With constant transmit powers, variations in cell radius affect the signal-to-interference-plus-noise ratio (SINR), which directly impacts interference detection capabilities. This simulation provides insights into the advantages of the proposed ZRD-SVDD in various cell sizes, particularly in noise-dominant environments, e.g., larger cell radii, where conventional IW schemes typically exhibit limited effectiveness.

Fig.~\ref{fig:f1score_tp} shows the F1 score as a function of cell radius under fixed time and frequency resources, i.e., $N_t = 100$ and $N_f = 12$. In environments with significant interference ranging from 100 $\mathrm{m}$ to 500 $\mathrm{m}$ cell radii, the proposed ZRD-SVDD exhibits performance comparable to certain baselines such as OC-SVM and $k$-NN with $k=20$. However, in macro cell environments with weak interference, the F1 score of the proposed scheme is lower than that of OC-SVM. To further analyze this result, we examine the two components of the F1 score: precision and sensitivity.

Figs.~\ref{fig:precision_tp} and \ref{fig:sensitivity_tp} present the precision and sensitivity performance under the same simulation conditions as in Fig.~\ref{fig:f1score_tp}. The proposed ZRD-SVDD demonstrates superior precision compared to the baselines but experiences a significant drop in sensitivity as the cell radius approaches 1 $\mathrm{km}$. This degradation in sensitivity is the primary factor contributing to the reduction in the F1 score, indicating a substantial increase in the rate of missed detections by ZRD-SVDD in larger cells. However, this missed detection is advantageous in practice since the proposed ZRD-SVDD predominantly treats weak interference as noise and tends to exhibit IW-off behavior for weak interference environments.
In contrast, $k$-NN suffers from a significant decrease in precision in weak interference environments, failing to properly detect interference. OC-SVM accurately detects most instances of interference but maintains consistently high sensitivity across all cell sizes. This results in an increased rate of false alarms in weak interference environments, misclassifying desired signals as interference and frequently triggering IW-on behavior. This analysis highlights the unique ability of ZRD-SVDD to effectively manage weak interference environments, distinguishing it from baselines that tend to overcompensate in such environments.

\begin{figure}[!t]
    \centering
    \includegraphics[width=\columnwidth]{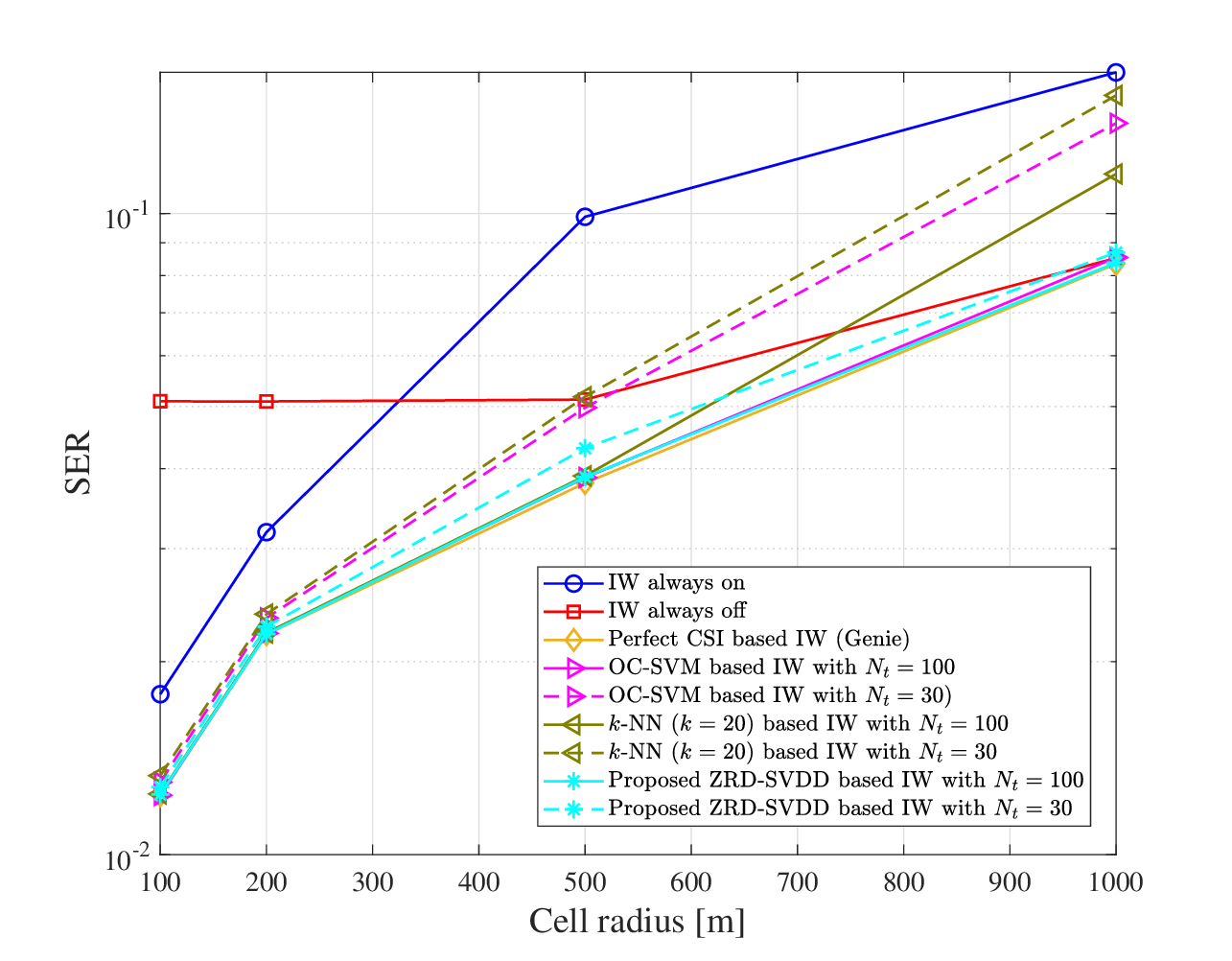}
    \caption{SER vs. cell radius: Comparison of $N_t=30$ and $N_t=100$ with $N_f=12$, $P_{\mathrm{S}} = P_{\mathrm{I}} = 46$ dBm.}
    \label{fig:ser timedomain}
\end{figure}

Figs.~\ref{fig:ser timedomain}~and~\ref{fig:ser freqdomain} illustrate the SER performance of the ZRD-SVDD-based IW scheme as a function of cell radius under limited time or frequency resources. For comparative analysis, we present the IW always-on and -off schemes, the genie-aided IW scheme, and the IW schemes based on baselines.  
As highlighted in Section~\ref{sec2.B}, while the IW always-on scheme can yield performance gains in strong interference environments, it may degrade the performance in weak interference environments due to limited time or frequency resources.

Fig. \ref{fig:ser timedomain} demonstrates the SER performance of the ZRD-SVDD-based IW scheme for $N_t=30$ and $100$ cases. In environments with cell radii between 100 and 200 $\mathrm{m}$, i.e., strong interference environments, all schemes exhibit performance equivalent to the genie-aided IW scheme. However, in weak interference environments, such as at the edge of macro cells, only the proposed scheme maintains performance comparable to the genie-aided IW scheme. This behavior is attributed to its similarity to the IW always-off scheme in weak interference environments.
Among the baselines, the OC-SVM-based IW scheme matches the performance of the genie-aided IW scheme across all regions when using 100 training samples. However, with only 30 training samples, the performance gap becomes significant. Our proposed scheme maintains its effectiveness, while the performance of the OC-SVM-based IW scheme deteriorates, eventually approaching that of the IW always-on scheme.

\begin{figure}[!t]
    \centering
    \includegraphics[width=.98\columnwidth]{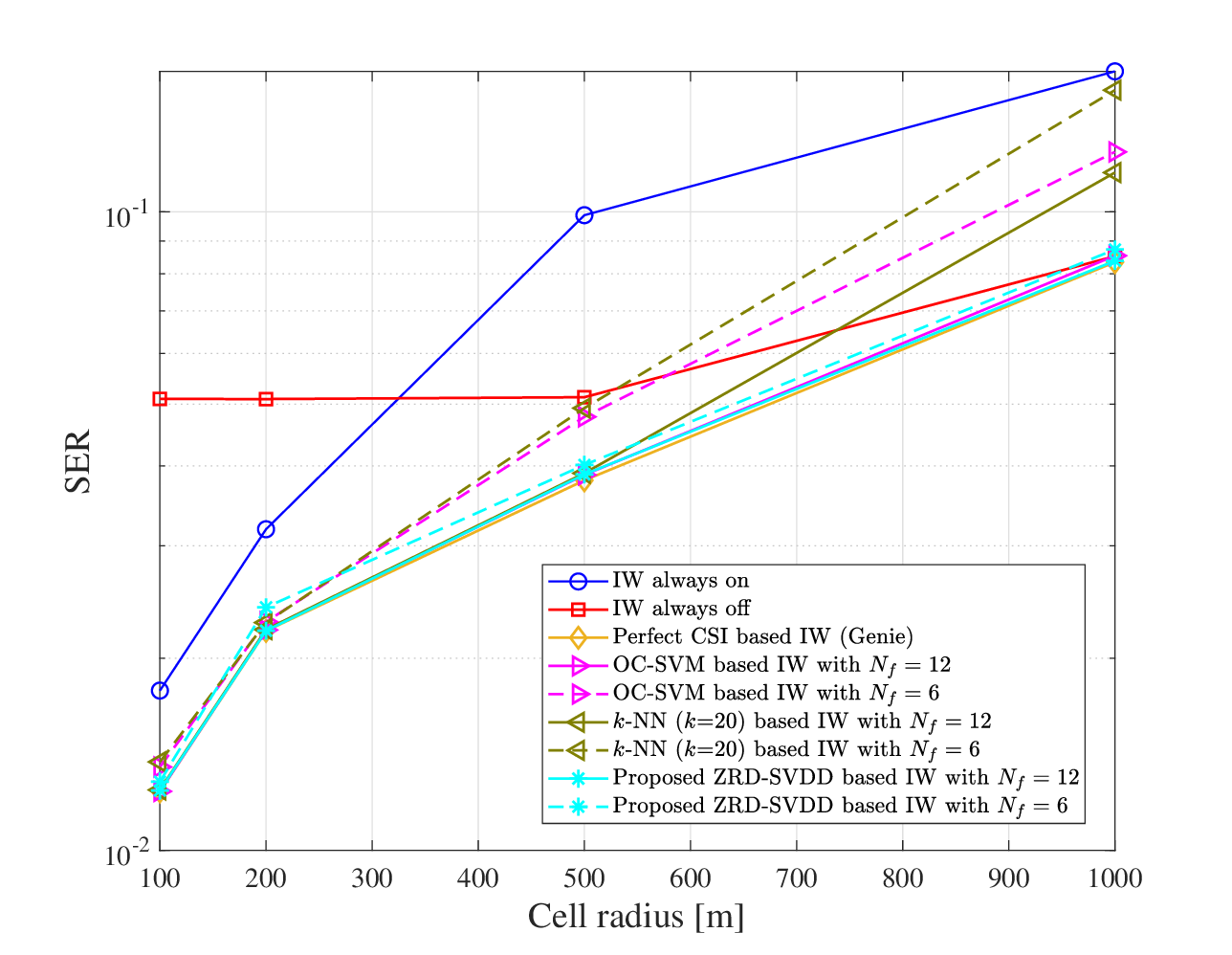}
    \caption{SER vs. cell radius: Comparison of $N_f=6$ and $N_f=12$ with $N_t = 100$, $P_{\mathrm{S}} = P_{\mathrm{I}} = 46$ dBm.}
    \label{fig:ser freqdomain}
\end{figure}

Fig. \ref{fig:ser freqdomain} presents the SER performance with $N_f=6$ and $12$ cases. Similar to Fig. \ref{fig:ser timedomain}, the ZRD-SVDD-based IW scheme demonstrates performance in line with the genie-aided IW scheme across the entire range, despite using fewer frequency resources. In contrast, baselines exhibit noticeable performance degradation when frequency resources are reduced from 12 to 6.
This superior performance of the ZRD-SVDD-based IW scheme indicates the potential for improved efficiency in practical systems.

Figs.~\ref{fig:TE BLER} and \ref{fig:TE TP} present the results of implementing the ZRD-SVDD scheme in a modem chipset, illustrating the PDSCH BLER and downlink throughput performance, respectively. In these experiments, the transmit power of the serving gNB\footnote{While 5G FR1 gNBs can transmit at up to 46~$\mathrm{dBm}$, our experimental setup uses approximately -80~$\mathrm{dBm}$. This lower power effectively compensates for the absence of signal attenuation typically experienced in cable connections and simulates the dynamic range of received signal strengths found in real wireless environments, thereby allowing for accurate emulation of field conditions.} $P_{\mathrm{S}}$ is varied from $-83~\mathrm{dBm}$ to $-77~\mathrm{dBm}$, while the interference power $P_{\mathrm{I}}$ ranges from $-99~\mathrm{dBm}$ to $-79~\mathrm{dBm}$.

As shown in Fig.~\ref{fig:TE BLER}, the performance of IW always-on and -off schemes varies with SINR levels across TDL-A, -B, and -C channel models. Notably, the proposed ZRD-SVDD scheme consistently performs better than these conventional IW schemes under all channel conditions. This TE experiment corroborates our numerical simulation results, further validating the effectiveness of the proposed ZRD-SVDD scheme.

Fig.~\ref{fig:TE TP} shows that the ZRD-SVDD-based IW scheme outperforms IW always-on and -off schemes in terms of downlink throughput regardless of channel models. The most notable performance difference is observed in the intermediate interference range, specifically the SINR of 0~$\mathrm{dB}$, when applying TDL-A and -B models. In this region, the ZRD-SVDD-based IW scheme demonstrates significant performance improvements over the conventional IW always-on and -off schemes.
Quantitatively, the throughput improvement ranges from 11\% to 21\% for the TDL-A model compared to both IW always-on and IW always-off schemes. Even more substantial gains are observed with the TDL-B model, where the performance enhancement ranges from 11\% to 32\% over the same baselines. These results highlight the adaptive nature of the ZRD-SVDD-based IW scheme, which effectively adjusts to varying interference environments.
By addressing a key limitation of conventional IW schemes, the ZRD-SVDD-based IW scheme enhances performance in challenging intermediate interference environments. This capability offers potential advantages in practical wireless communication systems.

\begin{figure}[!t]
    \centering
    \includegraphics[width=.98\columnwidth]{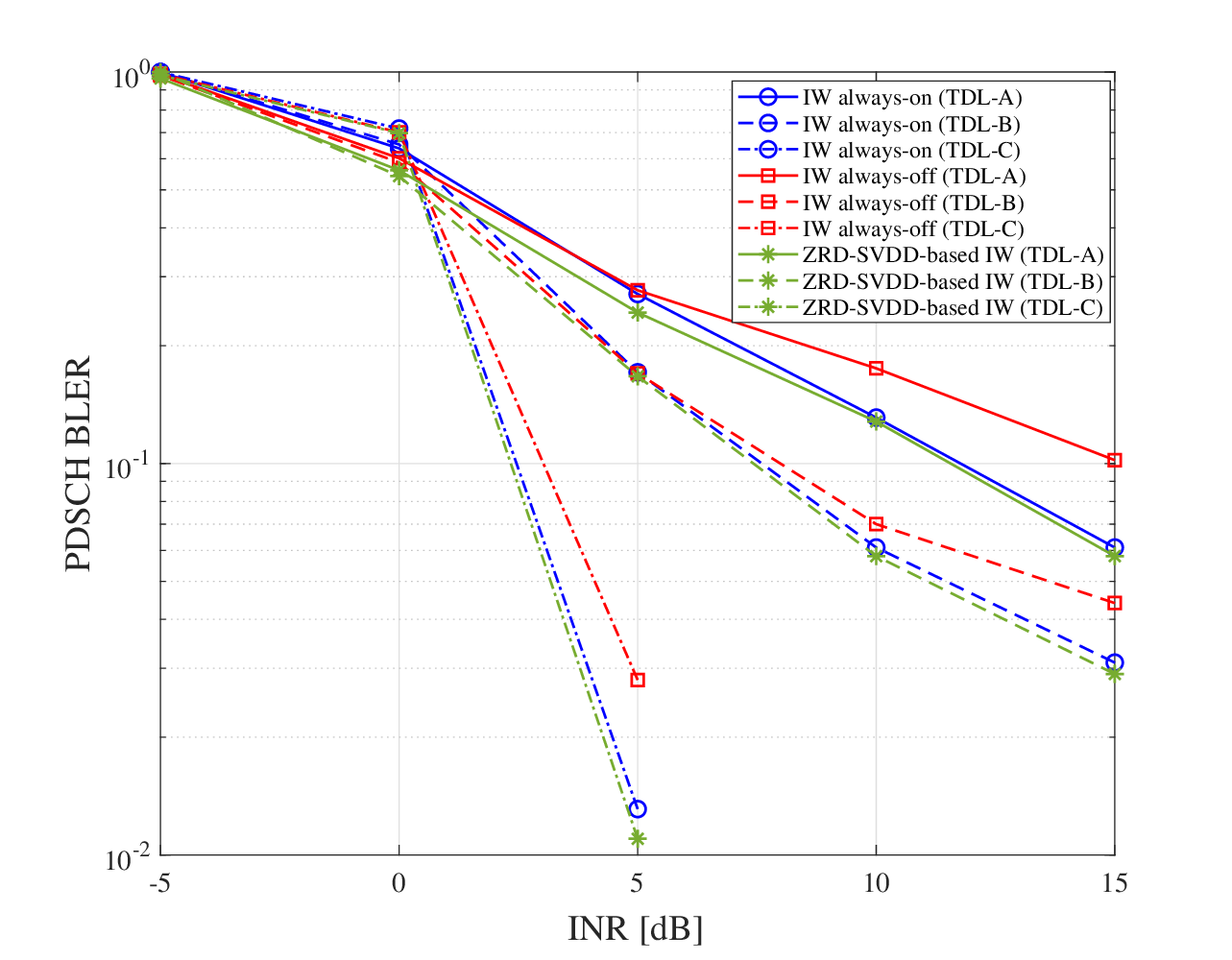}
    \caption{PDSCH BLER vs. INR with MCS~3 and Rank~1.}
    \label{fig:TE BLER}
\end{figure}

\section{Conclusions}
\label{sec6}
In this paper, we addressed the critical challenge of interference suppression in 5G NR systems, where increasing frequency reuse has intensified the impact of ICI. We proposed a novel UE-centric interference suppression scheme that enhances conventional IW by leveraging ZRD-SVDD, an effective anomaly detection technique combining Z-score normalization with deep learning-based SVDD.

Our comprehensive evaluation, including numerical simulations and TE experiments, demonstrated the significant benefits of the proposed ZRD-SVDD-based IW scheme. The results consistently showed that our proposed scheme outperforms various baselines, even under constrained time or frequency resources for training. In particular, the performance of our proposed scheme closely approaches that of the genie-aided IW scheme, which represents the performance upper bound.
Moreover, TE experiments conducted using 3GPP-defined channel models, i.e., TDL-A, -B, and -C, further validated the superior performance of our proposed scheme over conventional IW always-on and -off schemes.

These findings underscore the substantial potential of our UE-centric ZRD-SVDD-based IW scheme in mitigating ICI in 5G NR systems. This study establishes a robust foundation for advanced wireless communication systems, where interference suppression on the UE side will play a pivotal role in improving overall system performance.

\begin{figure}[!t]
    \centering
    \includegraphics[width=\columnwidth]{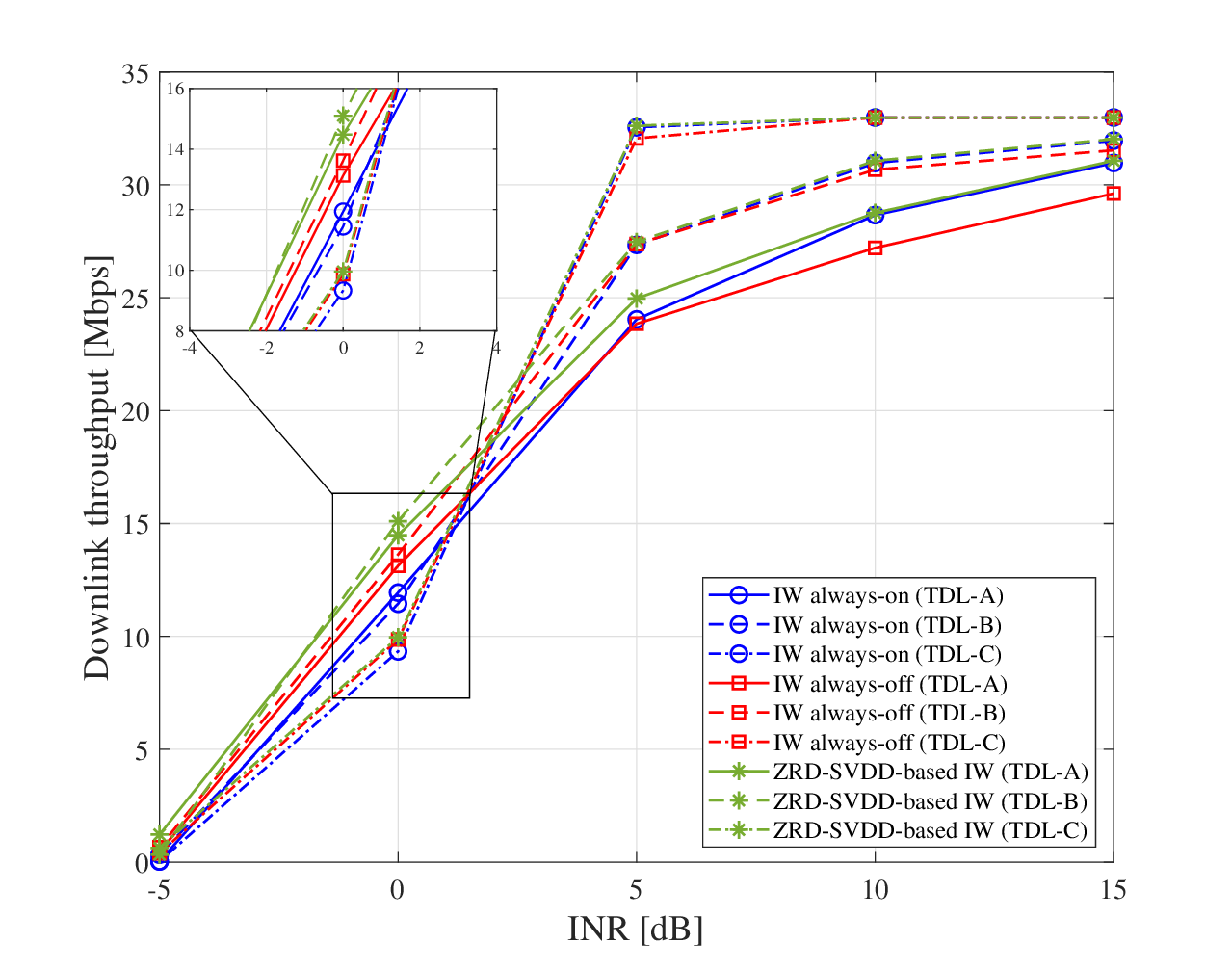}
    \caption{Downlink throughput vs. INR with MCS~3 and Rank~1.}
    \label{fig:TE TP}
\end{figure}

\section*{Appendix A \\ Proof of Theorem 1}

For a finite number of zero-mean square matrices \({\bA_t \in \mathbb{C}^{N_r\times N_r}}\), \(t= 1,\cdots, T_s\), which are independent and satisfy the condition \(\Vert \bA_t \Vert_2 \le L \) almost surely, we apply the matrix Bernstein inequality \cite{tropp2012user}, a type of concentration inequality, which is given as
\begin{align} \label{bernstein}
    \mathbb{P}\left( { \left\Vert  \sum_{t=1}^{T_s} \bA_t  \right \Vert_2} \ge \epsilon \right) \le 2N_r\exp\left(-\frac{\epsilon^2/2}{\sigma_{\mathrm{B}}^2 + {2L\epsilon}/{3}}\right),
\end{align}
where $\epsilon$ is a predefined parameter representing the estimation error tolerance level, $L$ is a positive constant denoting the upper bound of the spectral norm of $\bA_t$, and \( \sigma_{\mathrm{B}}^2 = \max\{\Vert \sum_t \mathbb{E} [\bA_t \bA_t^\mathrm{H}] \Vert_2 , \Vert \sum_t \mathbb{E} [\bA_t^\mathrm{H} \bA_t] \Vert_2\} \).

For a short duration, we can assume that the channel $\bG_m$ and transmit signal $\bz_m$ are fixed. 
The covariance of a channel is estimated using independent and identically distributed (i.i.d.) sample sequences \(\bu_{m,t} \bu_{m,t}^\mathrm{H}\), whose expectation equals the true covariance value. Based on this assumption, the true covariance matrix $\bR_m$ can be expressed as 
\begin{align}
    \bR_m = \mathbb{E}\left[{\bu_{m,t} \bu_{m,t}^\mathrm{H}}\right] = \bG_m \bz_m \bz_m^{\mathrm{H}} \bG_m^{\mathrm{H}} + \sigma_m^2 \bI_{N_r},
\end{align} for all $t =1,\cdots, T_s$.

To apply the matrix Bernstein inequality and guarantee the zero-mean condition, we define a new variable \(\bZ_t = \bu_{m,t} \bu_{m,t}^\mathrm{H} - \bR_m\). Consequently, from the result of  \eqref{bernstein}, the following inequality is satisfied as
\begin{align}
\label{diffSampleTrue}
    \mathbb{P}\left( { \left\Vert \sum_{t=1}^{T_s}\bZ_t  \right \Vert_2} \ge \epsilon \right)  \le 2N_r\exp\left(-\frac{\epsilon^2/2}{\sigma_{\mathrm{B}}^2 + {2L_z\epsilon}/{3}}\right),
\end{align}
where \(L_z\) denotes the upper bound that holds almost surely on the spectral norms of \(\bZ_t\), i.e., $\Vert \bZ_t \Vert_2 \leq L_z$.

The probability that the difference between the sample covariance and true covariance is within a specified error bound $\epsilon$ can be computed using a concentration inequality derived from \eqref{diffSampleTrue} as given below
\begin{align}
    &\mathbb{P}\left(  \frac{1}{{T_s}}{ \left\Vert \sum_{t=1}^{T_s} 
\left ( \bu_{m,t} \bu_{m,t}^\mathrm{H} - \bR_m \right ) \right \Vert_2}  < \epsilon  \right) \notag \\ 
    &=  \mathbb{P}\left( { \left\Vert \sum_{t=1}^{T_s}\bZ_t  \right \Vert_2} < \epsilon T_s  \right)  \notag\\    
    &= 1- \mathbb{P}\left( { \left\Vert \sum_{t=1}^{T_s}\bZ_t  \right \Vert_2} \ge \epsilon T_s  \right) \notag\\
     &> 1- 2N_r\exp\left(-\frac{\epsilon^2 T_s^2 /2}{\sigma_{\mathrm{B}}^2 + {2L_z \epsilon T_s}/{3}}\right) \notag \\
     &\stackrel{{(a)}}{\geq} 1- 2N_r\exp\left(-\frac{\epsilon^2 T_s^2 /2}{\sigma_{\mathrm{F}}^2 + {2L_z \epsilon T_s}/{3}}\right)
     \label{ineqZ},
\end{align}
where \( (a) \) follows the matrix norm property $\Vert \bA \Vert_2 \leq \Vert \bA \Vert_{\mathrm{F}}$, as described in \cite{golub2013matrix}. 

Now, we need to derive the variables $\sigma_{\mathrm{F}}^2$ and $L_z$ in \eqref{ineqZ} to finish the proof.
The term $\sigma_{\mathrm{F}}^2$ is expressed as
\begin{align}
\label{sigmaFrob}
\sigma_{\mathrm{F}}^2 &= \max\{\Vert \sum_t \mathbb{E} [\bZ_t \bZ_t^\mathrm{H}] \Vert_{\mathrm{F}} , \Vert \sum_t \mathbb{E} [\bZ_t^\mathrm{H} \bZ_t] \Vert_{\mathrm{F}}\} \notag \\
 &\stackrel{{(b)}}{=} \max \{ T_s\Vert \mathbb{E} [\bZ_t \bZ_t^\mathrm{H}] \Vert_{\mathrm{F}}, T_s \Vert \mathbb{E} [\bZ_t^\mathrm{H} \bZ_t] \Vert_{\mathrm{F}}\}\notag \\
    &\stackrel{{(c)}}{=} T_s \Vert \mathbb{E} [\bZ_t \bZ_t^\mathrm{H}] \Vert_{\mathrm{F}}.
\end{align}
In this derivation, \( (b) \) follows from the i.i.d. nature of the sample generation, while \( (c) \) holds due to the Hermitian property of \( \bZ_t\). To compute \(\mathbb{E} [\bZ_t \bZ_t^\mathrm{H}]\), we simplify $\bZ_t$ as
\begin{align}
\label{zoft}
\bZ_t = \bG_{m}\bz_m \bn_{m,t}^\mathrm{H} + \bn_{m,t}\bz_m^\mathrm{H}\bG_{m}^\mathrm{H} + \bn_{m,t}\bn_{m,t}^\mathrm{H} - \sigma_{m}^2 \bI_{N_r}.    
\end{align}
Based on \eqref{zoft}, $\mathbb{E} [\bZ_t \bZ_t^\mathrm{H}]$ is calculated using the following results,
\begin{align}
&\mathbb{E} [(\bG_{m}\bz_m \bn_{m,t}^\mathrm{H})(\bG_{m}\bz_m \bn_{m,t}^\mathrm{H})^\mathrm{H} ] = N_r\sigma_{m}^2 \bG_{m}\bz_m  \bz_m^\mathrm{H}\mathbf{G}_{m}^\mathrm{H}, \notag\\
&\mathbb{E} [(\bn_{m,t}\bz_m^\mathrm{H}\bG_{m}^\mathrm{H})(\bn_{m,t}\bz_m^\mathrm{H}\bG_{m}^\mathrm{H})^\mathrm{H} ] = \Vert \bG_m \bz_m \Vert_2^2\sigma_m^2\mathbf{I}_{N_r}, \notag\\
&\mathbb{E} [(\bn_{m,t} \bn_{m,t}^\mathrm{H})(\bn_{m,t} \bn_{m,t}^\mathrm{H})^\mathrm{H} ] = (N+1)\sigma_{m}^4 \bI_{N_r}, \notag \\
&\mathbb{E} [(\bn_{m,t} \bn_{m,t}^\mathrm{H}) \sigma_m^2 \bI_{N_r} ] =  \sigma_m^4{\mathbf{I}_{N_r}}.
\end{align}
Consequently, $\mathbb{E} [\bZ_t \bZ_t^\mathrm{H}]$ is derived as 
\begin{align} \label{final}
    \mathbb{E} [\bZ_t \bZ_t^\mathrm{H}] &= N_r \sigma_{m}^2 \bg_m\bg_m^\mathrm{H} + \Vert \bg_m \Vert_2^2 \sigma_m^2 \bI_{N_r}  + N_r \sigma_m^4 \bI_{N_r} \notag \\
    &= N_r \sigma_{m}^2 \bg_m\bg_m^\mathrm{H} + C_{1} \sigma_m^2 \bI_{N_r},
\end{align}
where \( \bg_m = \bG_m \bz_m\) and $C_{1} = \Vert \bg_m \Vert_2^2  + N_r \sigma_m^2 $. 

Next, we compute the Frobenius norm of $\mathbb{E} [\bZ_t \bZ_t^\mathrm{H}]$ using the matrix trace operation as follows
\begin{align}
    \Vert \mathbb{E} [\bZ_t \bZ_t^\mathrm{H}] \Vert_\mathrm{F}^2 = \mathrm{Tr}\left(\mathbb{E} [\bZ_t \bZ_t^\mathrm{H}] \mathbb{E} [\bZ_t \bZ_t^\mathrm{H}]^{\mathrm{H}} \right).
    \label{tracematrix}
\end{align}
This calculation can be decomposed into three components. For the first term, we calculate the trace as 
\begin{align}
    \mathrm{Tr} & \left(\left( N_r \sigma_m^2  \bg_m \bg_m^{\mathrm{H}} \right)  \left( N_r \sigma_m^2 \bg_m \bg_m^{\mathrm{H}} \right)^{\mathrm{H}} \right) \notag \\
    &= N_r^2 \sigma_m^4 \mathrm{Tr} \left( \bg_m \bg_m^{\mathrm{H}} \bg_m \bg_m^{\mathrm{H}} \right) = N_r^2 \sigma_m^4\Vert \bg_m  \Vert_2^4.
\end{align}
Subsequently, we compute the trace of the second term as follows
\begin{align}
    \mathrm{Tr}  \left( 2 C_1N_r \sigma_m^4 \bg_m \bg_m^{\mathrm{H}} \right) &= 2 C_1N_r \sigma_m^4 \mathrm{Tr}\left( \bg_m \bg_m^{\mathrm{H}} \right) \notag \\
    &= 2 C_1N_r \sigma_m^4 \Vert \bg_m  \Vert_2^2.
\end{align}
Finally, we evaluate the trace resulting from the third term, which involves the identity matrix
\begin{align}
    \mathrm{Tr} & \left( \left( C_1 \sigma_m^2 \bI_{N_r} \right) \left(  C_1 \sigma_m^2 \bI_{N_r} \right)^{\mathrm{H}} \right) \notag = C_1^2 \sigma_m^4 N_r.
\end{align}
By combining these three components, we obtain the final expression for $\sigma_{\mathrm{F}}^2$ as
\begin{align}
\label{sigma_F_2}
\sigma_{\mathrm{F}}^2 = T_s \sigma_m^2 \sqrt{N_r^2 \Vert \bg_m  \Vert_2^4 + 2 C_1N_r  \Vert \bg_m  \Vert_2^2 + C_1^2  N_r}.
\end{align}

To calculate $L_z$, we use the expression for $\bZ_t$ given in \eqref{zoft}. The upper bound for $\Vert\bZ_t \Vert_2$ can be derived as follows
\begin{align}
    \Vert \bZ_t \Vert_2 &\leq \mathbb{E}\left[\Vert\bg_{m}\bn_{m,t}^\mathrm{H} \Vert_2\right] + \mathbb{E}\left[\Vert \bn_{m,t}\bg_m^\mathrm{H} \Vert_2\right] + \mathbb{E}\left[\Vert \bn_{m,t}\bn_{m,t}^\mathrm{H} \Vert_2 \right] \notag \\
    & \enspace- \Vert \sigma_{m}^2 \bI_{N_r} \Vert_2 \notag \\
    &\approx 2 \sigma_m \sqrt{N_r} \Vert \bg_m \Vert_2 + \left( N_r - 1 \right) \sigma_m^2.
\end{align}
This approximation is based on the expected values of the norms of the noise terms. Consequently, we can approximate $L_z$ as
\begin{align}
\label{L_z}
    L_z \approx 2 \sigma_m \sqrt{N_r} \Vert \bg_m \Vert_2 + \left( N_r - 1 \right) \sigma_m^2,
\end{align}
which finishes the proof.

\bibliographystyle{IEEEtran}
\bibliography{refs_all}

\vfill
	
\end{document}